\documentclass[final,3p,times]{elsarticle}

\usepackage{graphicx,times,psfig}         
\usepackage{amssymb}
\usepackage{amsmath}
\usepackage{amsthm}
\usepackage{float}
\usepackage{epstopdf}
\usepackage{cases}
\usepackage{color}

\usepackage{subcaption}

\newtheorem{definition}{Definition}
\newtheorem{theorem}{\textbf{Theorem}}[section]
\newtheorem{lemma}[theorem]{\textbf{Lemma}}
\newtheorem{corollary}[theorem]{Corollary}
\newtheorem{assumption}[theorem]{Assumption}

\newtheorem{remark}{Remark}

\usepackage[labelsep=period]{caption}

\biboptions{sort&compress}

\begin{document}

\begin{frontmatter}
\title{A unified framework of fully distributed adaptive output time-varying formation control for linear multi-agent systems: an observer viewpoint}

\author[label0]{Wei Jiang}\ead{wjiang.lab@gmail.com}
\author[label1,label11,label111]{Zhaoxia Peng
}
 \ead{pengzhaoxia@buaa.edu.cn}
\author[label2]{Guoguang Wen}\ead{guoguang.wen@bjtu.edu.cn}
\author[label0]{Ahmed Rahmani}\ead{ahmed.rahmani@centralelille.fr}
\address[label0]{CRIStAL, UMR CNRS 9189, Ecole Centrale de Lille, Villeneuve d'Ascq, France}
\address[label1]{School of Transportation Science and Engineering, Beihang University, Beijing, 100191,P.R.China}
\address[label11]{Beijing Engineering Center for Clean Energy \& High Efficient Power, Beihang University, Beijing 100191, P.R.China}
\address[label111]{Laboratoire international associ$\acute{e}$, Beihang Universitym, P.R.China}
\address[label2]{Department of Mathematics, Beijing Jiaotong University, Beijing 100044, P.R.China}

%
%
%
\begin{abstract}
This paper presents a unified framework of time-varying formation (TVF) design for general linear multi-agent systems (MAS) based on an observer viewpoint from undirected to directed topology, from stabilization to tracking and from a leader without input to a one with bounded input.
The followers can form a TVF shape which is specified by piecewise continuously differential vectors. The leader's trajectory, which is available to only a subset of followers, is also time-varying.
For the undirected formation tracking and directed formation stabilization cases, only the relative output measurements of neighbors are required to design control protocols; for the directed formation tracking case, the agents need to be introspective (i.e. agents have partial knowledge of their own states) and the output measurements are required. Furthermore, considering the real applications, the leader with bounded input case is studied.
One main contribution of this paper is that fully distributed adaptive output protocols, which require no global information of communication topology and do not need the absolute or relative state information, are proposed to solve the TVF control problem.
Numerical simulations including an application to nonholonomic mobile vehicles are provided to verify the theoretical results.
\end{abstract}


\end{frontmatter}

\section{Introduction}
{V}{arious} cooperative control problems of the MAS have attracted much research interest in the past decades. This research domain includes consensus control \cite{olfati-saber_consensus_2004,ren_consensus_2005,yu_second-order_2017}, containment control \cite{ji_containment_2008,cheng_containment_2016} and formation control \cite{fax_information_2004, tanner_leader--formation_2004, dong_time-varying_2016-1}, etc. As one of the most important issues, formation control  has been paid much attention due to its broad potential applications like unmanned aerial vehicles \cite{wang_integrated_2013}, \cite{ghommam_three-dimensional_2016}, mobile robots \cite{antonelli_decentralized_2014,liu_distributed_2013, sakurama_multi-robot_2016,peng_leaderfollower_2013}, target enclosing \cite{zhang_cooperative_2016}, surveillance \cite{nigam_control_2012} and so on. Three formation control strategies, namely virtual structure \cite{lewis_high_1997}, leader-follower \cite{das_vision-based_2002} and behavior approaches \cite{balch_behavior-based_1998}, have been well studied.

Nowadays much more attention has been devoted to formation problems based on local information in a distributed way \cite{tanner_controllability_2004,wang_integrated_2013,liu_distributed_2013,peymani_almost_2014,sakurama_multi-robot_2016,oh_survey_2015}. The work in \cite{liu_distributed_2013}
used local relative position measurements of neighbors to achieve the leader follower formation control for unicycle robots. The unicycle model was transformed into a double-integrator dynamics and the small gain method was employed to deal with the nonholonomic constraint. However, all the follower robots need the leader's velocity and acceleration information which is a heavy communication burden. Peymani \cite{peymani_almost_2014} solved the $\mathcal{H}_{\infty}$ almost formation problem with the output regulation for general linear MAS while tracking a virtual reference at the same time. Agents were assumed to be introspective meaning that more sensors will be needed to test their own states as well as the relative outputs measurements.
In \cite{sakurama_multi-robot_2016} only the one-dimensional distance sensor information (i.e. relative positions of nearby vehicles) was used to design the distributed leader-follower formation controller for omni-directional vehicles.

The above-cited papers tackled the time-invariant formation control problems, but sometimes changing formation shapes is more necessary for two reasons: covering the greater area or avoiding collisions with obstacles. Time-varying formation (TVF) control, which was studied in \cite{dong_time-varying_2016-1}, \cite{antonelli_decentralized_2014} based on state information, means that the MAS can change formation shapes in certain circumstances and keep being stable simultaneously. 
Anonelli \cite{antonelli_decentralized_2014} proposed a distributed controller-observer schema for time-varying centroid and formation control of multi-robot systems with first-order dynamics. The proposed solution works for the strongly connected topology which is a more stringent constraint compared with the directed spanning tree topology. Moreover, each follower needs the knowledge of the number of all robots, which is a global information of the whole system, meaning that the protocol is not fully distributed.
In \cite{dong_time-varying_2016-1} the state information was used to solve the TVF problem with the switching directed spanning tree topologies.  However, the design of coupling strength parameter in the protocol is dependent on the minimal positive eigenvalue of Laplacian matrices subject to all the switching topologies, which is a global information for each agent, thus the control protocol is not fully distributed as well.
So designing a fully distributed protocol for TVF control problem is very necessary and quite practical in real applications.

Another reminder is that many previous references supposed that the state measurements could be utilized to design protocols. However, these measurements are sometimes unavailable in practice. Hence, solving the formation problem via an output method is quite important. Motivated by this observation, the output TVF control was studied based on the relative output measurements of neighbor agents in \cite{dong_output_2016}, but the design of parameter matrices are related to the eigenvalue information of the Laplacian matrix, which means not fully distributed again.

In the aforementioned works, only formation stabilization or maintenance problems were studied, but forming a formation is usually the first step in some real applications for the MAS. Actually there are some higher lever tasks such as enclosing a target or trajectory tracking. Then formation tracking problem arises in these scenarios.
A solution to the TVF tracking problem for collaborative heterogeneous MAS was presented in \cite{rahimi_time-varying_2014}, where a Lyapunov-based distributed controller was introduced based on virtual structure to make the whole system form a rigid formation among unmanned aerial vehicles and unmanned ground vehicles. However, the virtual leader needs to send its position and velocity information to all followers, which is a heavy communication cost that we should try to avoid.
Ghommam \cite{ghommam_three-dimensional_2016} studied the formation tracking control of multiple underactuated quadrotors when the reference signal is only available to a portion of quadrotors by using backstepping and filtering design techniques. The drawback is that the formation shape is unchangeable and what's more, the control protocol is still not fully distributed due to its dependence on communication topology information.
Another formation tracking research was done in \cite{dong_distributed_2016} where not only can the formation be time-varying but also so is the leader's trajectory. However, the agent dynamics can be only described as second-order model and the protocol, which is based on relative state information, is not fully distributed as well.

It is worth remarking that undirected or directed communication topology is also a key point to design local control laws which are aimed at achieving multi-agent formation stabilization or tracking control. Under the undirected topology, The authors in \cite{wang_integrated_2013}, \cite{ghommam_three-dimensional_2016} solved
the time-invariant formation tracking problem and the author in~\cite{wang2017distributed} addressed the TVF control problem in a fully distributed fashion. However, the undirected topology means that the information exchange among agents is bidirectional, which may consume much more communication and energy resources than a directed one. Furthermore, there are many applications in reality where information only flows in one direction. For instance, in the leader-follower MAS, the leader may become the only one equipped with a communication transmitter \cite{ren_consensus_2005}. The TVF stabilization and tracking problems with directed topology were studied in \cite{dong_time-varying_2016-1} and \cite{dong_distributed_2016}, recpectively.  The former work is not fully distributed and the latter one only deals with second order dynamics with a not fully distributed protocol.

Motivated by the above analysis,  we intend to design control protocols in a fully distributed way to solve the TVF control problem for general linear MAS based on output measurements, from undirected to directed topology, from leaderless to leader-follower, from the leader of zero input to the one of bounded input.
The proposed protocols in this paper are consensus-based formation controllers. Consensus problems have been widely investigated in \cite{olfati-saber_consensus_2004}, \cite{ren_consensus_2005}, \cite{jadbabaie_coordination_2003}, etc. The framework of consensus problems in networks of agents was done in \cite{olfati-saber_consensus_2004} and then, Ren \cite{ren_consensus_2007} proved that the virtual structure, leader-follower and behavior based formation methods could be unified in the framework of consensus-based methods. A significant breakthrough was presented in \cite{li_consensus_2011} where a unified framework was introduced to solve the consensus problem and the synchronization of complex networks. It expanded conventional observers to distributed observers by allowing agents to exchange state information which is estimated by local observers via the network. For more details about consensus design of continuous-time MAS, the survey paper \cite{cao_overview_2013} and references therein are recommended.

The results in this paper are based on our previous work~\cite{jiang2017distributed} which deals with the TVF tracking problem under undirected topology. The control protocols are independent of any global information, rely on agent dynamics and only the output measurements, thereby are fully distributed. Both the formation shapes and the leader's trajectory can be time-varying. 
The main contributions can be summarized as follows: 
\begin{itemize}
	\item A unified framework of TVF control design from the distributed observer viewpoint is presented. We reveal how to design observers to tackle TVF control problems from undirected to directed topology, from stabilization to tracking and from a leader without input to a one with bounded input.
	\item The protocols of majority of existing works are not fully distributed, where the protocol parameters are related to the minimal positive eigenvalue of the Laplacian matrix~\cite{dong_time-varying_2016-1,ghommam_three-dimensional_2016,dong_output_2016} or every follower needs the knowledge of the number of all robots~\cite{antonelli_decentralized_2014}. In this study, we design the protocols in a fully distributed fashion.
	\item Compared with most literature dealing with time-invariant formation control with first order dynamics \cite{sakurama_multi-robot_2016}, second order dynamics \cite{wang_integrated_2013}, \cite{liu_distributed_2013}, general linear dynamics \cite{peymani_almost_2014} or TVF control with first order dynamics \cite{antonelli_decentralized_2014}, this paper studies the TVF control with general linear dynamics.
	\item The output measurements, which are more applicable in real industry than the state ones~\cite{dong_time-varying_2016-1,antonelli_decentralized_2014} that are sometimes unavailable in reality, are utilized in this paper. Different from \cite{liu_distributed_2013,rahimi_time-varying_2014} where all the followers need to know the leader's information, in this paper, only a small portion of followers need the leader's information, which can reduce communication cost greatly especially in the case of large number of followers. 
\end{itemize}

\section{Preliminaries and model formulation}\label{section2}

\subsection{Mathematical preliminaries}

The connections between agents can be represented by a weighted graph $ \mathcal{G} = (\mathcal{V,E,A})$, where $\mathcal{V}$ and $\mathcal{E}$ denote the nodes and edges, respectively. $\mathcal{A} = [ a_{ij} ] \in \mathbb{R} ^{N \times N}$ denotes the adjacency matrix where $a_{ij}=1 $ if there exists a path from agent $ j $ to agent $ i $, and $a_{ij} = 0$ otherwise. An edge $\left(i,j\right) \in \mathcal{E}$ in graph $ \mathcal{G}$ means that agent $j$ can receive information from agent $i$ but not necessarily conversely.
The Laplacian matrix $\mathcal{L} = [ l_{ij} ] \in \mathbb{R} ^{N \times N}$ is normally defined as $l_{ii}= \sum_{j \neq i} a_{ij} $ and $l_{ij}= -a_{ij} $ when $ i \neq j$. 
A graph is said to be undirected if $\left(i,j\right) \in \mathcal{E}$  implies $\left(j,i\right) \in \mathcal{E}$ for any $i,j \in \mathcal{V}$. 
An undirected graph is connected if there exists a path between each pair of distinct nodes. 
A directed path from node $i$ to $j$ is a sequence of edges $\left(i,i_{1}\right),\left(i_{1},i_{2}\right), \ldots,\left(i_{k},j\right)$ with different nodes $i_{s}, s=1,2,\ldots, k$. A digraph (i.e., directed graph) is strongly connected if there is a directed path from each node to each other node. 
A digraph contains a directed spanning tree if there is a node from which a directed path exists to each other node.
A digraph has a directed spanning tree if it is strongly connected, but not vice versa.
More graph theories can be found in~\cite{godsil_algebraic_2001}.

The symbol $ \textbf{1} $ denotes a column vector with all entries being 1.
Matrix dimensions are supposed to be compatible if not explicitly stated.
The symbol $\otimes$ represents the Kronecker product and $ diag\{a_1, \ldots, a_{n}\} $ denotes a diagonal matrix with the diagonal entries being $  a_1, \ldots, a_{n}$.
The matrix $ A=[a_{ij}] \in \mathbb{R}^{N\times N} $ is called a nonsingular $ M $-matrix if $ a_{ij} \le 0, \forall i\neq j $, and all eigenvalues of $ A $ have positive real parts. Here, $ \lambda_{min}(A) $ and $ \lambda_{max} (A)$ represent the minimal and maximal eigenvalues of $ A $, respectively.

\begin{lemma}[\cite{ren_consensus_2005}]\label{lemma_ren2005}
	The Laplacian matrix $ \mathcal{L} $ of a directed communication topology $ \mathcal{G} $ has at least one zero eigenvalue with $\textbf{1}$ as a right eigenvector, and has all nonzero eigenvalues with positive real parts. Furthurmore, zero is a simple eigenvalue of $ \mathcal{L} $ if and only if $ \mathcal{G} $ contains a directed spanning tree.
\end{lemma}
\begin{lemma}[\cite{qu_cooperative_2009}]\label{lemma_M_matrix}
	For a nonsingular $ M $-matrix $ A $, there exists a positive diagonal matrix $ G = diag(g_{1}, \ldots, g_{N}) > 0 $ such that $ GA+A^{T}G > 0 $.
\end{lemma}
\begin{lemma}[\cite{mei_consensus_2014}]\label{lemma_stronglyconnected}
	Suppose that $ \mathcal{G} $ is strongly connected. Then there exists a positive vector $ r=[r_{1}^{T}, \ldots , r_{N}^{T}] >0 $ such that $ r^{T}\mathcal{L} =0 $, and $ \hat{\mathcal{L}} = R \mathcal{L}+\mathcal{L}R$ is the symmetric Laplacian matrix associated with an undirected connected graph where $ R = diag(r_{1}, \ldots, r_{N}) $. Moreover, $ \min_{\chi ^{T}x=0,x \neq 0} \frac{x^{T}\hat{\mathcal{L}}x}{x^{T}x} > \frac{\lambda_{2}(\hat{\mathcal{L}})}{N} $, where $ \lambda_{2}(\hat{\mathcal{L}}) $ denotes the algebraic connectivity of $ \hat{\mathcal{L}} $, i,e., the smallest positive eigenvalue of $ \hat{\mathcal{L}} $, and $ \chi $ is any vector with positive entries.
\end{lemma}
\begin{lemma}[\cite{bernstein_matrix_2009}]\label{lemma_pq}
	If $ a, b $ are nonnegative real numbers and $ p, q $ are positive real numbers satisfying $ \frac{1}{p}+\frac{1}{q}=1 $, then $ ab \leq \frac{a^{p}}{p}+\frac{b^{q}}{q} $.
\end{lemma}

\subsection{Model formulation}

The identical general linear dynamics of agent $ i $ in the MAS is
\begin{equation}\label{eq:dynamics}
\begin{array}{rcl}
\dot x_{i} & = & Ax_{i}+Bu_{i},\\
y_{i} & = & Cx_{i}, \quad i =0, \ldots, N
\end{array}
\end{equation}
where $x_{i}=[x_{i1}, \ldots, x_{in}]^{T} \in \mathbb{R} ^{n}$ , $u_{i} \in \mathbb{R} ^{p}$ and $y_{i} \in \mathbb{R} ^{q}$ are the state, control input and measured output, respectively. $ A \in \mathbb{R} ^{n \times n}, B \in \mathbb{R} ^{n \times p} $ and $ C \in \mathbb{R} ^{q \times n} $ are constant matrices.
\begin{assumption}\label{assumptionABC}
	$ (A,B,C) $ is stabilizable and detectable.
\end{assumption}

Without loss of generality, suppose that agents in \eqref{eq:dynamics} indexed by $ 1, \ldots, N $ are the followers denoted as $ \mathbb{F} = \left\lbrace 1,\ldots,N\right\rbrace  $ and the agent indexed by $ 0 $ is the leader whose output information is only available to a small portion of followers.
Moreover, the leader  does not receive any information from the followers.

In the following, first the leader is regarded without control input, i.e., $ u_{0}=0 $,  which is a common assumption in many existing works on the distributed cooperative control of linear MAS \cite{wen_containment_2016,li_consensus_2013,lv_novel_2016}. 

However, as we know, where the whole system moves is decided by the leader and that is why the leader exists. Then where will the leader move? The answer is that a desired dynamic trajectory command is given to the leader to ask the leader to finish the desired trajectory tracking or that the leader moves anywhere it could, which requires the leader's control input to be nonzero.
Furthermore, $ u_{0}=0 $ means the leader is a virtual one and the desired trajectory has severe limitations because of the equation $ \dot x_{0}(t)  =  Ax_{0}(t) $ as  the system matrix $ A $ is unchangeable.
In real applications, the leader needs to regulate the final consensus trajectory. So its control input $ u_{0} $ will not be affected by followers. In this paper, we deal with the consensus control in a fully distributed fashion, which means $ u_{0} $ will not be accessible to any follower. This is more difficult than the case of $ u_{0}=0 $.
The leader with bounded control input $ u_{0}(t)$. will be presented in Section~\ref{section_bounded_input}.

Denote the TVF shape information for followers as 
$h\left(t\right)=[ h_{1}\left(t\right)^{T}, \ldots, h_{N}\left(t\right)^{T} ]^{T} \in \mathbb{R} ^{Nn}$ with $h_{i} (t) $ being piecewise continuously differentiable
\begin{equation}\label{h_dynamic}
\dot h_{i} (t) =  (A+BK_{1}) h_{i}(t)
\end{equation}
where $ K_{1} $ is a constant matrix to be designed. Designing $ K_{1} $ give us a freedom to design any TVF shape satisfying the equation \eqref{h_dynamic}. The detailed explanation of \eqref{h_dynamic} can be referred to our previous paper~\cite{jiang2017distributed}.

\begin{definition}\label{definitionformation}
	The system (\ref{eq:dynamics})  is said to achieve the output TVF tracking control if for any given initial states $ x_{i}(0), i \in \mathbb{F} $, there exists 
	\begin{equation}\label{eq:time-varying}
	\lim_{t \to \infty} \| y_{i}(t)-y_{0}(t)-Ch_{i}(t) \|=0.
	\end{equation}
\end{definition}

In this paper, each follower can get access to the relative output measurements which are expressed as
\begin{equation}\label{relative output}
\begin{aligned}
y_{ij} = y_{i}-y_{j}, \, y_{i0} = y_{i}-y_{0},  \quad i,j \in \mathbb{F}.
\end{aligned}
\end{equation} 

\section{Main results}\label{mainresults}
This section mainly focuses on how to design fully distributed adaptive protocols to address TVF control problems from undirected to directed topology, from formation stabilization to tracking and from a leader of no input to a one with bounded input. In a word, a unified framework of protocol designing is presented.

Section~\ref{sub_undirected} solved the formation stabilization and tracking problems under undirected topology based only on relative output measurements. After that, we modified the protocol of Section~\ref{sub_undirected} to address the same problem under directed topology in Section~\ref{undirected_direacted}. Unfortunately, the designed protocol is not perfect since every follower needs to know the leader's output, which is a heavy communication burden for the whole system. In order to relax this constraint and make the control effect perfect, the protocol design for TVF tracking problem under directed topology is divided into three steps. Firstly, the TVF stabilization problem under directed topology is solved in Section~\ref{sec_stabilizaiton} in which we do not take the leader into consideration. Then in Section~\ref{section_tracking}, we tackle the TVF tracking problem under directed topology with a leader of no control input. Finally, the extended case of a leader with bounded input is studied in Section~\ref{section_bounded_input}.

\subsection{Undirected formation tracking}\label{sub_undirected}

\begin{assumption}\label{assumptiontopology}
	The communication subgraph $ \tilde{\mathcal{G}} $ among followers is undirected with the adjacency matrix $ \tilde{\mathcal{A}} $. The graph $ \mathcal{G} $ of the whole system contains a spanning tree with the adjacency matrix $ \mathcal{A} $ where the leader acts as the root node. 
\end{assumption}

The objective here is to design the fully distributed protocol to make followers form the TVF shape and track the leader simultaneously based only on relative outputs under the undirected communication topology. To do this, the protocol for each follower $ i $ is proposed as
\begin{equation}\label{eq:protocol}
\begin{aligned}
u_{i} =& K_{1}h_{i}+K_{2}v_{i},\\
\dot v_{i} =& (A+BK_{2})v_{i}+F\left[  \sum_{j=1}^{N} a_{ij}c_{ij}(\bar{c}_{ij} -y_{ij})+d_{i}c_{i}(\bar{c}_{i} -y_{i0}) \right] ,\\
\dot c_{ij} =& k_{ij}a_{ij} (\bar{c}_{ij} -y_{ij})^{T} \Gamma (\bar{c}_{ij} -y_{ij}),\\
\dot c_{i} =& k_{i}d_{i}(\bar{c}_{i} -y_{i0})^{T} \Gamma (\bar{c}_{i} -y_{i0}), i \in \mathbb{F}
\end{aligned}
\end{equation}
where $ \bar{c}_{i}=C(v_{i}+h_{i}), \bar{c}_{ij}=\bar{c}_{i}-\bar{c}_{j}  $. $ K_{1}, K_{2}$ are the feedback gain matrices. $v_{i} \in \mathbb{R}^{n} $ is the observer state of follower $ i $ and $ a_{ij} $ is the $ (i,j) $-th entry of adjacency matrix $ \tilde{\mathcal{A}} $. $ c_{ij}(t) $ denotes the time-varying coupling weight between follower $ i $ and $ j $ with $ c_{ij}(0) = c_{ji}(0) \geq 0 $ and $ c_{i} \geq 0 $ denotes the coupling weight between follower $ i $ and the leader. $ k_{ij}=k_{ji}, k_{i} $ are positive constants and $ F \in \mathbb{R}^{n \times q}, \Gamma \in \mathbb{R}^{q \times q} $ are the feedback gain matrices to be determined. $ d_{i} $ satisfies $ d_{i}=1 $ if follower $ i $ can get information from the leader, otherwise $ d_{i}=0 $.

\begin{remark}
	The adaptive coupling weights $ c_{ij}(t)$ and $c_{i}(t) $ can release the constraint that some protocols need to know the minimal positive eigenvalue of the Laplacian matrix $ \mathcal{L} $, e.g., 
	in \cite{dong_time-varying_2016-1}, \cite{dong_distributed_2016}, \cite{wen_containment_2016}. In other words, $ c_{ij}(t)$ and $c_{i}(t) $ can make the protocol fully distributed under the undirected communication topology.
\end{remark}


\begin{theorem}\label{theorem_undrected}
	The fully distributed TVF tracking problem is solved with Assumptions \ref{assumptionABC} and \ref{assumptiontopology} under the protocol \eqref{eq:protocol} if $ A+BK_{2} $ is Hurwitz, $ \Gamma = I $ and $ F=-PC^{T} $, where  $ P^{-1} > 0 $ is a solution to the following linear matrix inequality (LMI)
	\begin{equation}\label{lmi}
	P^{-1} A+A^{T}P^{-1}-2C^{T}C <0.
	\end{equation}
	Moreover, the coupling weights $ c_{ij}(t), c_{i}(t), i,j \in \mathbb{F} $ converge to some finite steady-state values.
\end{theorem}
\begin{proof}
	The proving and analyzing detail is referred to our previous work~\cite{jiang2017distributed}.
\end{proof}

For the special case where there is no leader $ (i.e., d_{i}=0, i \in \mathbb{F}) $, the communication graph $ \mathcal{G} $ in Assumption~\ref{assumptiontopology} becomes undirected and connected.
\begin{assumption}\label{assum_stabilization_undirected}
	The graph $ \mathcal{G} $ among agents is undirected and connected.
\end{assumption}

\begin{assumption}\label{assum_AB}
	(A,B) is stabilizable.
\end{assumption}

In the following, we give the definition for formation stabilization problem.
\begin{definition}\label{definitionformation_keeping}
	The MAS (\ref{eq:dynamics}) is said to achieve the output TVF stabilization if for any given initial states $ x_{i}(0), i \in \mathbb{F} $, there exists 
	\begin{equation}\label{eq:time-varyingnotracking}
	\lim_{t \to \infty} \| (y_{i}(t)-y_{j}(t))-C(h_{i}(t)-h_{j}(t)) \|=0.
	\end{equation}
\end{definition}

Then we present the following fully distributed adaptive protocol to address the TVF stabilization problem by using only the relative output information
\begin{equation}\label{eq:protocolstate}
\begin{aligned}
u_{i} =& K_{1}h_{i}+K_{2}v_{i},\\
\dot v_{i} =& (A+BK_{2})v_{i}+F   \sum_{j=1}^{N} a_{ij}c_{ij}(\bar{c}_{ij} -y_{ij}),\\
\dot c_{ij} =& k_{ij}a_{ij} (\bar{c}_{ij} -y_{ij})^{T} \Gamma (\bar{c}_{ij} -y_{ij}), i \in \mathbb{F}.
\end{aligned}
\end{equation}

\begin{corollary}
	The fully distributed TVF stabilization problem is solved with Assumptions \ref{assum_stabilization_undirected} and \ref{assum_AB} under the protocol \eqref{eq:protocolstate} if $ A+BK_{2} $ is Hurwitz, $ \Gamma = I $ and $ F=-PC^{T} $, where  $ P^{-1} > 0 $ is a solution to the LMI \eqref{lmi}.
	And $ c_{ij}(t), i,j \in \mathbb{F}$ converge to some finite steady-state values.
\end{corollary}
The proof here is similar and thus is omitted for conciseness.

\subsection{Directed formation tracking with full access to the leader}\label{undirected_direacted}

\begin{assumption}\label{assumptiondirected}
	The graph $ \mathcal{G} $ contains a directed spanning tree where the leader acts as the root node.
\end{assumption}    

As we know, each follower has access to a weighted linear combination of relative outputs between itself and its neighbors. The network measurement for follower $ i $ can be synthesized as a single signal
\begin{equation}\label{network measurement}
\tilde{y}_{ij0}= \sum_{j=1}^{N} a_{ij}y_{ij}+d_{i}y_{i0}, i \in \mathbb{F}
\end{equation}
where $ y_{ij}$ and $ y_{i0} $ are defined in~\eqref{relative output}.
In Section \ref{sub_undirected} the key point to solve the formation tracking problem is $c_{ij}(t)=c_{ji}(t) $ in protocol \eqref{eq:protocol} due to the property of undirected topology $ \tilde{\mathcal{G}} $ in Assumption~\ref{assumptiontopology}, namely, $a_{ij}=a_{ji} $. But for directed topology, the adjacency matrix $ \mathcal{A} $ does not have the symmetric property, i.e. $a_{ij}\neq a_{ji} $. So the parameter $c_{ij}(t) $, which is the time-varying coupling weight between follower $ i $ and $ j $, will not be suitable for the protocol design of directed topology. In addition, note that parameter $ c_{i}(t) $ denotes the coupling weight between follower $ i $ and the leader. Based on the above finding, we replace the parameters $ c_{ij}(t) $ and $ c_{i}(t) $ by one parameter $ c_{i}(t) \geq 0 $ denoting the time-varying coupling weight associated with the $ i $-th follower, and modify the protocol \eqref{eq:protocol} to a new one in the following form
\begin{equation}\label{eq:protocol_directed_notgood}
\begin{aligned}
u_{i} =& K_{1}h_{i}+K_{2}v_{i},\\
\dot v_{i} =& (A+BK_{2})v_{i}+F (c_{i}+\rho_{i})\left(   \sum_{j=1}^{N} a_{ij}\bar{c}_{ij}+d_{i}\bar{c}_{i} -  \tilde{y}_{ij0} \right) \\
\dot c_{i} =& (\bar{c}_{i}-y_{i0})^{T} \Gamma (\bar{c}_{i}-y_{i0}), i \in \mathbb{F}
\end{aligned}
\end{equation}
where  $ \bar{c}_{i}=C(v_{i}+h_{i}), \bar{c}_{ij}=\bar{c}_{i}-\bar{c}_{j}  $ and $ \rho_{i} $ is a smooth and monotonically increasing function to be determined later. Other parameters are the same as in \eqref{eq:protocol}. 

Define the TVF tracking error $ \tilde x_{i}=x_{i}-h_{i}-x_{0} $ and the observer error $ e_{i}=\tilde{x}_{i}-v_{i}, e=\left[ e_{1}^{T}, \ldots , e_{N}^{T} \right] ^{T} $. 
we give the following theorem to solve the TVF tracking problem under directed topology.

\begin{theorem}\label{theorem_drected_notgood}
	The fully distributed TVF tracking problem is solved with Assumptions \ref{assumptionABC} and \ref{assumptiondirected} under the protocol \eqref{eq:protocol_directed_notgood} if $ A+BK_{2} $ is Hurwitz, $ \Gamma = I $, $ F=-P C^{T} $ and $ \rho_{i}=e_{i}P^{-1}e_{i} $, where $ P ^{-1} > 0 $ is a solution to the LMI~\eqref{lmi}.
	Moreover, the coupling weight $ c_{i} (t), i \in \mathbb{F}$ converge to some finite steady-state values.
\end{theorem}

\begin{proof}
	From $ y_{i}(t)-y_{0}(t)-Ch_{i}(t)=C \tilde{x}_{i}$, it follows \eqref{eq:time-varying} that the objective is to prove $ \lim_{t \to \infty} \tilde{x}_{i}=0, i \in \mathbb{F} $. Using \eqref{eq:dynamics}, \eqref{h_dynamic} and \eqref{eq:protocol_directed_notgood}, the TVF tracking error $ \tilde{x}_{i} $ satisfies
	\begin{equation}\label{tildex}
	\dot{\tilde{x}}_{i}=A \tilde{x}_{i}+BK_{2}v_{i}=(A+BK_{2})\tilde{x}-BK_{2}e_{i}.
	\end{equation}
	The objective now changes to prove $ \lim_{t \to \infty}e_{i}=0 $ such that $ \lim_{t \to \infty} \tilde{x}_{i}=0, i \in \mathbb{F} $ since $ A+BK_{2} $ is Hurwitz. 
	Using \eqref{tildex} and \eqref{eq:protocol_directed_notgood}, the system \eqref{eq:dynamics} can be rewritten in the following form
	\begin{equation}\label{e_dynamic_notgood}
	\begin{array}{rcl}
	\dot e_{i} &=& Ae_{i}+FC(c_{i}+\rho_{i})( \sum_{j=1}^{N} l_{ij}e_{i}+d_{i}e_{i}), \\[2ex]
	\dot c_{i} &=& e_{i}^{T} C^{T} \Gamma C e_{i}.
	\end{array}
	\end{equation}
	
	Let
	\begin{equation}\label{v2}
	V_{1}=\frac{1}{2}\sum_{i=1}^{N} g_{i}(2c_{i}+\rho_{i})\rho_{i}+\frac{1}{2}\sum_{i=1}^{N} g_{i}(c_{i}- \alpha)^{2}
	\end{equation}
	where $ g_{i} >0, i \in \mathbb{F} $ is defined in Lemma~\ref{lemma_M_matrix}.
	It follows from $ c_{i}(0) >0 $ and $\dot{c}_{i}(t)>0 $ that $ c_{i}(t) >0, \forall t>0 $.
	$ \alpha $ is a positive constant to be determined. Noting further that $ \rho_{i} \geq 0 $, thus $ V_{1} $ is positive definite. Then
	\begin{equation}\label{L_derivative_notgood}
	\begin{aligned}
	\dot{V}_{1} 
	=& \displaystyle \sum_{i=1}^{N} [g_{i}(c_{i}+\rho_{i}) \dot{\rho}_{i} +g_{i}\rho_{i}\dot{c}_{i}+g_{i}(c_{i}-\alpha)\dot{c}_{i}]\\
	=& e^{T} [G(\hat{c}+\hat{\rho}) \otimes(P^{-1}A+A^{T}P^{-1})  +(\hat{c}+\hat{\rho})
	(G\hat{\mathcal{L}} + \hat{\mathcal{L}}^{T}G)(\hat{c}+\hat{\rho}) \otimes P^{-1}FC + G(\hat{c}+\hat{\rho} - \alpha I) \otimes C^{T} \Gamma C  ] e \\
	\le & e^{T} [G(\hat{c}+\hat{\rho}) \otimes(P^{-1}A+A^{T}P^{-1})  -\lambda_{0}(\hat{c}+\hat{\rho})^{2} \otimes C^{T}C+ G(\hat{c}+\hat{\rho} - \alpha I) \otimes C^{T} C ] e
	\end{aligned}
	\end{equation}
	where  $ \hat{c} = diag(c_{1}, \ldots, c_{N}), \hat{\rho} = diag(\rho_{1},\ldots,\rho_{N}) $, $ \mathcal{D} = diag\{d_{1}, \ldots, d_{N}\} $ and $ \hat{\mathcal{L}}= \mathcal{L}+\mathcal{D}$. The Laplacian matrix $\mathcal{L} $ is corresponding to the subgraph among followers. It is known that $ \mathcal{D} > 0 $ with at least one diagonal entry being positive since at least one follower can get information from the leader. Then $ \hat{\mathcal{L}} $ is a $ M $-matrix with graph $ \mathcal{G} $ satisfying Assumption~\ref{assumptiondirected}, which means all eigenvalues of $ \hat{\mathcal{L}} $ have positive real parts \cite{hong_tracking_2006}. Furthermore, from Lemma \ref{lemma_M_matrix} there exists
	$ G = diag(g_{1}, \ldots, g_{N}) > 0 $ such that $ G\hat{\mathcal{L}}+\hat{\mathcal{L}}^{T}G \ge \lambda_{0} I $ where $ \lambda_{0} $ is the smallest positive eigenvalue of $ G\hat{\mathcal{L}}+\hat{\mathcal{L}}^{T}G $. Using Lemma \ref{lemma_pq} we get
	\begin{equation}\label{pq_notgood}
	e^{T}[G(\hat{c}+\hat{\rho}) \otimes C^{T} C] e \le e^{T}[(\frac{\lambda_{0}}{2} (\hat{c}+\hat{\rho})^{2} + \frac{G^{2}}{2\lambda_{0}}) \otimes C^{T} C]e.
	\end{equation}
	Substituting \eqref{pq_notgood} into \eqref{L_derivative_notgood} we have
	\begin{equation}
	\begin{aligned}
	\dot{V}_{3} \le & e^{T} [G(\hat{c}+\hat{\rho}) \otimes(P^{-1}A+A^{T}P^{-1})  - (\frac{\lambda_{0}}{2} (\hat{c}+\hat{\rho})^{2} - \frac{G^{2}}{2\lambda_{0}} +\alpha G) \otimes C^{T} C]e.
	\end{aligned} 
	\end{equation}
	Choosing $ \alpha \ge \max_{i \in \mathbb{F}} \frac{5g_{i}}{\sqrt{2\lambda_{0}}} $, we obtain
	\begin{equation}
	\dot{V}_{1} \le  e^{T} [G(\hat{c}+\hat{\rho}) \otimes \mathcal{X}]e
	\le  0
	\end{equation}
	where the last inequality comes directly from the LMI \eqref{lmi} which is $ \mathcal{X}=P^{-1}A+A^{T}P^{-1} -2C^{T} C <0 $. 
	Then $ V_{1}(t) $ is bounded and so is $ c_{i}(t) $. Each coupling weight $ c_{i}(t) $ increases monotonically and converges to some finite value finally. Note that $ \dot{V}_{1}\equiv 0 $ is equivalent to $ e =0 $. By LaSalle's Invariance principle \cite{krstic_nonlinear_1995}, it follows that $ \lim_{t \to \infty} e_{i}=0 $ such that $ v_{i} \rightarrow \tilde{x}_{i} $ as $ t \rightarrow \infty $, which means the function of  each follower's distributed observer $ v_{i} $ in \eqref{eq:protocol_directed_notgood} is to estimate its own TVF tracking error $ \tilde{x}_{i} $. 
	
	Since $ A+BK_{2} $ is Hurwitz and $ \lim_{t \to \infty} e_{i}=0 $, from \eqref{tildex} we have $ \lim_{t \to \infty} \tilde{x}_{i}=0, i \in \mathbb{F} $, i.e., the distributed  TVF tracking problem under the directed topology satisfying Assumption \ref{assumptiondirected} is solved.
\end{proof}

\begin{remark}
	Note that in order to calculate $ c_{i}(t) $ in protocol \eqref{eq:protocol_directed_notgood}, each follower $ i $ requires the knowledge of $ y_{i0} $ in \eqref{relative output}, namely the relative output measurement between the follower $ i $ and the leader. It means every follower needs to know the leader's output information, in other words, $ d_{i}>0, \forall i \in \mathbb{F} $, which is a stringent communication constraint and will increase communication cost heavily. We will solve the TVF tracking problem where the leader's output information is only available to a small subset of followers in the following sections.
\end{remark}

\subsection{Directed formation stabilization}\label{sec_stabilizaiton}

\begin{assumption}\label{assumption_stabilization}
	The communication graph $ \mathcal{G} $ is strongly connected.
\end{assumption} 

In Section \ref{undirected_direacted} we solved the leader-follower TVF tracking problem with directed spanning tree topology, but it requires each follower to know the leader's output information. In order to relax this severe constraint, we start to solve the formation stabilization problem first, namely without the leader. The inspiration comes from the last section. Recall the equation of adaptive parameter $ c_{i} $ in protocol \eqref{eq:protocol_directed_notgood} of Section \ref{undirected_direacted} as
\begin{equation*}
\dot c_{i} = (\tilde{x}_{i}-v_{i})^{T} C^{T} \Gamma C (\tilde{x}_{i}-v_{i}).
\end{equation*}

\begin{remark}
	It is obvious that the observer $ v_{i} $ in \eqref{eq:protocol_directed_notgood} is used to estimate the formation tracking error $ \tilde{x}_{i} = x_{i} -h_{i} -x_{0} $. For the formation stabilization problem without a leader, it is natural to design $ v_{i} $ to estimate $ \bar{x}_{i} = x_{i} -h_{i} $.
\end{remark}

Our goal in this section is to design the fully distributed protocol based only on the relative output measurements to make the system form a shape, namely, making agents $ i $ and $ j $ satisfy formation stabilization definition \eqref{eq:time-varyingnotracking}. Similar as the network measurement \eqref{network measurement}, we define two signals as
\begin{equation}\label{psi_eta}
\psi_{i}=\sum_{j=1}^{N} a_{ij}(v_{i}-v_{j}), \,
\eta_{i}=\sum_{j=1}^{N} a_{ij}(\bar{x}_{i}-\bar{x}_{j}).
\end{equation}
Denote $ \psi=[\psi_{1}^{T}, \ldots , \psi_{N}^{T}]^{T},  \eta=[\eta_{1}^{T}, \ldots , \eta_{N}^{T}]^{T}$, then $ \eta = (\mathcal{L} \otimes I_{n})\bar{x} $, where $ \mathcal{L} $ is the Laplacian matrix corresponding to the graph $ \mathcal{G} $ satisfying Assumption~\ref{assumption_stabilization}. Under Lemma~\ref{lemma_ren2005}, $ \mathcal{L} $ has a zero eigenvalue and other eigenvalues with positive real parts. From Definition \ref{definitionformation_keeping} we can say that the TVF stabilization problem is solved if $ \lim_{t \to \infty} \eta =0 $. So $ \eta $ can be viewed as formation stabilization error in this section. 

The fully distributed adaptive protocol based only on relative output measurements is proposed for each agent $ i $ as
\begin{equation}\label{eq:protocol_directed_stabilization}
\begin{aligned}
u_{i} =& K_{1}h_{i}+K_{2}v_{i},\\
\dot v_{i} =& (A+BK_{2})v_{i}+F (c_{i}+\rho_{i}) \sum_{j=1}^{N} a_{ij}(\bar{c}_{ij}-y_{ij}), \\
\dot c_{i} =& \left[ \sum_{j=1}^{N} a_{ij}(\bar{c}_{ij}-y_{ij})\right] ^{T}  \Gamma \sum_{j=1}^{N} a_{ij}(\bar{c}_{ij}-y_{ij})
\end{aligned}
\end{equation}
where $   \bar{c}_{i}=C(v_{i}+h_{i}), \bar{c}_{ij}=\bar{c}_{i}-\bar{c}_{j}  $ and other parameters are defined similarly as protocol \eqref{eq:protocol_directed_notgood}. By substituting \eqref{psi_eta} into \eqref{eq:protocol_directed_stabilization} we can write protocol \eqref{eq:protocol_directed_stabilization} as
\begin{equation}
\begin{aligned}
u_{i} =& K_{1}h_{i}+K_{2}v_{i},\\
\dot v_{i} =& (A+BK_{2})v_{i}+FC (c_{i}+\rho_{i}) (\psi_{i}-\eta_{i}), \\
\dot c_{i} =& (\psi_{i}-\eta_{i})^{T} C^{T} \Gamma C (\psi_{i}-\eta_{i}).
\end{aligned}
\end{equation}

Note that the term $ C (\psi_{i}-\eta_{i}) $ implies that each agent needs to receive the virtual outputs $ Cv_{j} $ and $ C\bar{x}_{j} $ from its neighbors via the communication graph $ \mathcal{G} $ satisfying Assumption~\ref{assumption_stabilization}. Let $ \varrho_{i} = \psi_{i}- \eta_{i}, \varrho= [\varrho_{1}^{T}, \ldots , \varrho_{N}^{T}]^{T}$, then we combine \eqref{psi_eta}, \eqref{eq:protocol_directed_stabilization}, \eqref{eq:dynamics} and get
\begin{equation}\label{psi_eta_varrho}
\begin{aligned}
\dot{\psi} =& [I_{N} \otimes (A+BK_{2})] \psi +[\mathcal{L} (\hat{c}+\hat{\rho}) \otimes FC]\varrho,\\
\dot{\eta}=& [I_{N} \otimes (A+BK_{2})] \eta +(I_{N} \otimes BK_{2})\varrho,\\
\dot{\varrho} =&  [I_{N} \otimes A + \mathcal{L} (\hat{c}+\hat{\rho}) \otimes FC]\varrho.
\end{aligned}
\end{equation}


\begin{theorem}\label{theorem_stabilization}
	Suppose Assumptions \ref{assumptionABC} and \ref{assumption_stabilization} hold, the fully distributed TVF stabilization problem is solved under the protocol \eqref{eq:protocol_directed_stabilization} if $ A+BK_{2} $ is Hurwitz, $ \Gamma = I $, $ F=-Q^{-1} C^{T} $ and $ \rho_{i}=\varrho_{i}^{T}Q\varrho_{i}$, where $ Q > 0 $ is a solution to the LMI
	\begin{equation}\label{lmi_stabilication}
	Q A+A^{T}Q-2C^{T}C <0.
	\end{equation}
	And $ c_{i} (t) $ converge to some finite steady-state values.
\end{theorem}

\begin{proof}
	First, we prove that $ \lim_{t \to \infty}\varrho =0 $. To this end, similar as \eqref{v2} in Theorem \ref{theorem_drected_notgood}, let
	\begin{equation}\label{v3}
	V_{2}=\frac{1}{2}\sum_{i=1}^{N} r_{i}(2c_{i}+\rho_{i})\rho_{i}+\frac{1}{2}\sum_{i=1}^{N} r_{i}(c_{i}- \alpha)^{2}
	\end{equation}
	where $ r= [r_{1}^{T}, \ldots , r_{N}^{T}]^{T}, r_{i} >0 $ is the left eigenvector of $ \mathcal{L} $ associated with the zero eigenvalue and other parameters are the same as in Theorem \ref{theorem_drected_notgood}. Similarly, the derivative of $ V_{2} $ is
	\begin{equation}\label{v3_derivative}
	\begin{aligned}
	\dot{V}_{2} 
	=& \varrho^{T} [R(\hat{c}+\hat{\rho}) \otimes(QA+A^{T}Q) + R(\hat{c}+\hat{\rho} - \alpha I)  \otimes C^{T} \Gamma C +(\hat{c}+\hat{\rho})\tilde{\mathcal{L}}(\hat{c}+\hat{\rho}) \otimes QFC ] \varrho
	\end{aligned}
	\end{equation}
	where $ R = diag(r_{1},\ldots, r_{N}) $ and $ \tilde{\mathcal{L}} =  R\mathcal{L} + \mathcal{L}^{T}R$. Denote $ \tilde{\varrho} = [(\hat{c}+\hat{\rho}) \otimes I_{n}]\varrho$. Considering $ \varrho= \psi - \eta= (\mathcal{L} \otimes I_{n})(v-\bar{x}), r^{T}\mathcal{L}=0$, then
	\begin{equation*}
	\tilde{\varrho}^{T}[(\hat{c} +\hat{\rho})^{-1}r \otimes \textbf{1}]= (v-\bar{x})^{T}(L^{T}r \otimes \textbf{1}) =0.
	\end{equation*}
	Since each entry of $ r$ is positive, then each entry of $ [(\hat{c} +\hat{\rho})^{-1}r \otimes \textbf{1}] $ is also positive. From Lemma \ref{lemma_stronglyconnected}
	\begin{equation*}
	\begin{aligned}
	\tilde{\varrho}^{T} (\tilde{\mathcal{L}} \otimes I_{n})\tilde{\varrho} >  \frac{\lambda_{2}(\tilde{\mathcal{L}})}{N} \tilde{\varrho}^{T}\tilde{\varrho}
	= \frac{\lambda_{2}(\tilde{\mathcal{L}})}{N} \varrho^{T} [(\hat{c} + \hat{\rho})^{2} \otimes I_{n}]\varrho.
	\end{aligned}
	\end{equation*}
	Similar as \eqref{pq_notgood} in Theorem \ref{theorem_drected_notgood}, using Lemma \ref{lemma_pq} we get
	\begin{equation*}
	\begin{aligned}
	\varrho^{T}[R(\hat{c}+\hat{\rho}) \otimes C^{T} C] \varrho \le& \varrho^{T} [(\frac{\lambda_{2}(\tilde{\mathcal{L}})}{2N}(\hat{c}+\hat{\rho})^{2}+\frac{N}{2\lambda_{2}(\tilde{\mathcal{L}})}R^{2}) \otimes C^{T} C] \varrho.
	\end{aligned}
	\end{equation*}
	Combining above two inequalities with \eqref{v3_derivative} and choosing $ \alpha \ge \frac{5N\lambda_{max}(R)}{2\lambda_{2}(\tilde{\mathcal{L}})} $, we have
	\begin{equation}
	\begin{aligned}
	\dot{V}_{2} 
	\le & \varrho^{T} [R(\hat{c}+\hat{\rho}) \otimes(QA+A^{T}Q-2C^{T}C) ]\varrho\\
	\le & 0
	\end{aligned}
	\end{equation}
	where the last inequality comes from LMI \eqref{lmi_stabilication}. So $ V_{2}(t) $ is bounded, and $ c_{i}(t) $ increases monotonically and converges to some finite value finally. 
	Similar as the proof in Theorem~\ref{theorem_drected_notgood}, $ \lim_{t \to \infty} \varrho=0$ can be proved. 
	From the second equation in \eqref{psi_eta_varrho} and $ A+BK_{2} $ is Hurwitz, we can prove $ \lim_{t \to \infty} \eta=0 $, i.e., the fully distributed TVF stabilization problem with the directed strongly connected topology is solved.
\end{proof}

\begin{remark}
	From \eqref{psi_eta} it is easy to get $ \psi = (\mathcal{L} \otimes I_{n})v, \eta = (\mathcal{L} \otimes I_{n})\bar{x} $ and $ \varrho = (\mathcal{L} \otimes I_{n})(v-\bar{x}) $. $ \lim_{t \to \infty} \varrho=0$ means that the error between observer $ v_{i} $ and formation stabilization error $ \bar{x}_{i} $ of each agent $ i $ will go to zero eventually. Similarly, $ \lim_{t \to \infty} \eta=0$ means that the formation stabilization error $ \bar{x}_{i} $ of each agent $ i $ will reach consistent eventually. Obviously, the observer $ v_{i} $ of each agent $ i $ will also reach consistent eventually. Note that \eqref{eq:protocol_directed_stabilization} is a consensus-based formation stabilization protocol. From Corollary 1 of \cite{olfati-saber_consensus_2004}, we know that the group decision value of formation is a function of each agent's initial state $ x_{i}(0), i=1,\ldots,N $. The group decision value decides where the leaderless formation to go, which means there is no precisely explicit equation defining where the leaderless formation to go. It is necessary and applicable to solve the leader-follower TVF tracking problem with directed topology when only a small subset of followers know leader's output information, which will be presented in next section.
\end{remark}

For the special case where the relative state measurements $ x_{ij}= x_{i}-x_{j}, i,j = 1,\ldots,N $ are available among neighbors, the fully distributed adaptive protocol is proposed as
\begin{equation}\label{eq:protocol_directed_stabilization_state}
\begin{aligned}
u_{i} =& K_{1}h_{i}+K_{2}v_{i},\\
\dot v_{i} =& (A+BK_{2})v_{i}+B\tilde{F}(c_{i}+\rho_{i})  \sum_{j=1}^{N} a_{ij}(\hat{c}_{ij}-x_{ij}), \\
\dot c_{i} =& \left[  \sum_{j=1}^{N} a_{ij}(\hat{c}_{ij}-x_{ij}) \right] ^{T}  \Gamma \sum_{j=1}^{N} a_{ij}(\hat{c}_{ij}-x_{ij})
\end{aligned}
\end{equation}
where $  \hat{c}_{ij}=v_{i}+h_{i} -v_{j}-h_{j} $.

\begin{corollary}\label{corollary_stabilization}
	Suppose Assumptions \ref{assum_AB} and \ref{assumption_stabilization} hold, the fully distributed TVF stabilization problem is solved under the protocol \eqref{eq:protocol_directed_stabilization_state} if $ A+BK_{2} $ is Hurwitz, $ \Gamma = \tilde{Q}^{-1}BB^{T}\tilde{Q}^{-1} $, $ \tilde{F}=-B^{T}\tilde{Q}^{-1} $ and $ \rho_{i}=\varrho_{i}^{T}\tilde{Q}^{-1}\varrho_{i}$, where $ \tilde{Q} > 0 $ satisfies the following LMI
	\begin{equation}\label{lmi_stabilication_state}
	A\tilde{Q}+\tilde{Q}A^{T}-2BB^{T} <0.
	\end{equation}
	And $ c_{i}(t) $ converge to some finite steady-state values.
\end{corollary}
The proof is similar as in Theorem \ref{theorem_stabilization} and the details are omitted here.

\subsection{Directed formation tracking with partial access to the leader}\label{section_tracking}

During the process of solving the TVF stabilization problem with directed topology in Section \ref{sec_stabilizaiton}, an observer $ v_{i} $ that estimates formation stabilization error is introduced to design the protocol \eqref{eq:protocol_directed_stabilization} based on the following structure
$$
\left.
\begin{array}{l}
\displaystyle \lim_{t \to \infty} \varrho_{i}=0 \Rightarrow v_{i}-\bar{x}_{i} \rightarrow v_{j}-\bar{x}_{j}\\
\displaystyle \lim_{t \to \infty} \psi_{i}=0 \Rightarrow v_{i} \rightarrow v_{j}
\end{array}
\right\}
\Rightarrow  \bar{x}_{i} \rightarrow \bar{x}_{j}, t \rightarrow \infty.
$$

In this section for the formation tracking problem, similar to that structure, we introduce two observers to design the fully distributed protocol as follows
\begin{equation}\label{eq:protocol_directed_tracking}
\begin{aligned}
u_{i} =& K_{1}h_{i}+K_{2}v_{i},\\
\dot w_{i} =& Aw_{i} + Bu_{i} - BK_{1}h_{i} +F[Cw_{i} -(y_{i}-Ch_{i})],\\
\dot v_{i} =& Av_{i}+ Bu_{i} - BK_{1}h_{i}  +FC (c_{i}+\rho_{i})(\psi_{i}-\eta_{i}) +F[Cw_{i} -(y_{i}-Ch_{i})],\\
\dot c_{i} =& (\psi_{i}-\eta_{i})^{T} C^{T} \Gamma C (\psi_{i}-\eta_{i}), i \in \mathbb{F}.
\end{aligned}
\end{equation}
Here $ \psi_{i}= \sum_{j=0}^{N} a_{ij}(v_{i}-v_{j}), \eta_{i} = \sum_{j=0}^{N} a_{ij}(w_{i}-w_{j}) $, which is  similar to \eqref{psi_eta}.
And $ w_{0}=Aw_{0}+F(Cw_{0}-y_{0}), v_{0}=0 $ meaning that the leader has only one observer $ w_{0} $ to estimate its state $ x_{0} $. 
Note here that $ a_{i0}>0 $ means follower $ i $ can receive information from the leader and can not if $ a_{i0}=0 $, which shows that only a subset of followers can get the leader's output information.
The local observer $ w_{i} $ is designed to estimate the formation stabilization error $ \bar{x}_{i} = x_{i}-h_{i} $ of each follower $ i $, while the distributed observer $ v_{i} $ is used to make formation tracking error $ \tilde{x}_{i}=\bar{x}_{i}-x_{0} $ converge to zero. Here we assume that each agent is introspective as termed in \cite{peymani_almost_2014}, which means each one has access to its own output. 

Under Assumption \ref{assumptiondirected}, the Laplacian matrix of graph $ \mathcal{G} $ can be partitioned as $ \mathcal{L} =  
\begin{bmatrix}
0 & 0_{1 \times N} \\
\mathcal{L}_{2} & \mathcal{L}_{1} 
\end{bmatrix} $, where $ \mathcal{L}_{1} \in \mathbb{R}^{N \times N} ,  \mathcal{L}_{2} \in \mathbb{R}^{N \times 1}$. It is easy to confirm that $ \mathcal{L}_{1} $ is a nonsingular $ M $-matrix. 

Denote $ w=[w_{1}^{T}, \ldots, w_{N}^{T}]^{T}, v=[v_{1}^{T}, \ldots, v_{N}^{T}]^{T} $ and $ \varrho_{i}=\psi_{i}-\eta_{i} , i \in \mathbb{F} $, then
\begin{equation}\label{psi_1}
\begin{aligned}
\psi=&( \mathcal{L}_{1} \otimes I_{n})v,\\
\eta=&(\mathcal{L}_{1} \otimes I_{n})(w- \textbf{1} \otimes w_{0}),\\
\varrho=&(\mathcal{L}_{1} \otimes I_{n})(v-w+\textbf{1} \otimes w_{0}).\\
\end{aligned}
\end{equation}

Our goal is try to prove that
$$
\left.
\begin{array}{r}
\left.
\begin{array}{l}
\varrho_{i}=0 \Rightarrow w_{i}-v_{i} \rightarrow w_{0}\\
\psi_{i}=0 \Rightarrow v_{i} \rightarrow 0
\end{array}
\right\}
\Rightarrow  w_{i} \rightarrow w_{0}\\
w_{i} \rightarrow \bar{x}_{i}, \quad w_{0} \rightarrow x_{0}
\end{array}
\right\}
\Rightarrow  \bar{x}_{i} \rightarrow x_{0}
$$
where $ \bar{x}_{i}-x_{0}=x_{i}-h_{i}-x_{0} $ is the same as formation tracking error $ \tilde{x}_{i} $ in the proof of Theorem \ref{theorem_drected_notgood}.

In this section, similar as~\eqref{network measurement}, define a signal as
\begin{equation}
\hat{x}_{i}=\sum_{j=1}^{N} a_{ij}(\bar{x}_{i}-\bar{x}_{j})+a_{i0}(\bar{x}_{i}-x_{0})
\end{equation}
and $ \hat{x}=[\hat{x}_{1}^{T}, \ldots, \hat{x}_{N}^{T}]^{T} $, then $ \hat{x}=(\mathcal{L}_{1} \otimes I_{n})(\bar{x}- \textbf{1} \otimes x_{0}) $. It is easy to see that the TVF tracking problem with the directed topology is solved if and only if $ \lim_{t \to \infty} \hat{x}=0 $. Substituting \eqref{eq:protocol_directed_tracking}, \eqref{psi_1} into \eqref{eq:dynamics}, we get
\begin{equation}\label{psi_eta_tracking}
\begin{aligned}
\dot{\hat{x}}=& (I_{N} \otimes A)\hat{x} +(I_{N} \otimes BK_{2})\psi,\\
\dot{\eta} =& (I_{N} \otimes A)\eta +(I_{N} \otimes BK_{2})\psi + (I_{N} \otimes FC)(\eta - \hat{x}),\\
\dot{\psi}=& [I_{N} \otimes (A+BK_{2})] \psi + [\mathcal{L}_{1}(\hat{c}+\hat{\rho}) \otimes FC]\varrho + (I_{N} \otimes FC)(\eta - \hat{x}) +(\mathcal{L}_{1} \otimes FC) [\textbf{1} \otimes (w_{0}-x_{0})],\\
\dot c_{i} =& (\psi_{i}-\eta_{i})^{T} C^{T} \Gamma C (\psi_{i}-\eta_{i}), \quad i \in \mathbb{F}.
\end{aligned}
\end{equation}
Defining $ \bar{x}_{0}=w_{0}-x_{0} $ as the leader's state estimation error, $ \zeta=[\zeta_{1}^{T}, \ldots, \zeta_{N}^{T}]^{T}=\eta-\hat{x} $ and $ \varrho=\psi -\eta $, we obtain
\begin{equation}\label{zeta}
\begin{aligned}
\dot{\zeta}=& [I_{N} \otimes (A+FC)]\zeta,\\
\dot{\varrho} =&  [I_{N} \otimes A + \mathcal{L}_{1}(\hat{c}+\hat{\rho}) \otimes FC]\varrho + (\mathcal{L}_{1}\otimes FC)(\textbf{1}\otimes \bar{x}_{0}),\\
\dot c_{i} =& \varrho_{i}^{T} C^{T} \Gamma C \varrho_{i}, \quad i \in \mathbb{F}.
\end{aligned}
\end{equation}

The following theorem presents a result of designing protocol \eqref{eq:protocol_directed_tracking} to solve the TVF tracking problem with only a small subset of followers knowing the leader's output information.
\begin{theorem}\label{theorem_tracking}
	The fully distributed TVF tracking problem is solved with Assumptions \ref{assumptionABC} and \ref{assumptiondirected} under the protocol \eqref{eq:protocol_directed_tracking} if $ A+BK_{2} $ is Hurwitz, $ \Gamma = I $, $ F=-Q^{-1} C^{T} $ and $ \rho_{i}=\varrho_{i}^{T}Q\varrho_{i}$, where $ Q > 0 $ satisfies the LMI \eqref{lmi_stabilication}.
	And $ c_{i}(t), i \in \mathbb{F} $ converge to some finite steady-state values.
\end{theorem}

\begin{proof}
	First, we prove that $ \lim_{t \to \infty}\varrho =0$ and $ \lim_{t \to \infty}\bar{x}_{0} =0  $. To this end, let
	\begin{equation}\label{v5}
	V_{3}=\frac{1}{2}\sum_{i=1}^{N} g_{i}(2c_{i}+\rho_{i})\rho_{i}+\frac{1}{2}\sum_{i=1}^{N} g_{i}(c_{i}- \alpha)^{2} + \gamma \bar{x}_{0}^{T}Q\bar{x}_{0}
	\end{equation}
	where $ \gamma $ is a positive constant to be determined later and other parameters are the same as in the proof of Theorem \ref{theorem_drected_notgood}. Similarly, $ V_{3} $ is positive definite with respect to $ \varrho_{i}, c_{i} $ and $ \bar{x}_{0} $. Then
	\begin{equation}
	\begin{aligned}\label{v5_derivative}
	\dot{V}_{3} 
	\le & \varrho^{T} [G(\hat{c}+\hat{\rho}) \otimes(QA+A^{T}Q) -\lambda_{0}^{'}(\hat{c}+\hat{\rho})^{2} \otimes C^{T}C + G(\hat{c}+\hat{\rho} - \alpha I) \otimes C^{T} C ] \varrho - \gamma \bar{x}_{0}^{T}W\bar{x}_{0}\\
	&-2\varrho^{T}[G(\hat{c}+\hat{\rho})\mathcal{L}_{1}\otimes C^{T}C ](\textbf{1}\otimes \bar{x}_{0})
	\end{aligned}
	\end{equation}
	where $ \lambda_{0}^{'} > 0 $ is the smallest eigenvalue of $ G\mathcal{L}_{1}+\mathcal{L}_{1}^{T}G $ and $ W=-(Q A+A^{T}Q-2C^{T}C) $ is a positive definite matrix according to \eqref{lmi_stabilication}. By using Lemma \ref{lemma_pq} and $ \mathcal{L}_{1}\textbf{1}=-\mathcal{L}_{2} $, we can get
	\begin{equation*}
	\varrho^{T}[G(\hat{c}+\hat{\rho}) \otimes C^{T} C] \varrho \le  \varrho^{T}[(\frac{\lambda_{0}^{'}}{3} (\hat{c}+\hat{\rho})^{2} + \frac{3G^{2}}{4\lambda_{0}}) \otimes C^{T} C]\varrho
	\end{equation*}
	and
	\begin{equation*}
	\begin{aligned}
	-2\varrho^{T}[G(\hat{c}+\hat{\rho})\mathcal{L}_{1}\otimes C^{T}C ](\textbf{1}\otimes \bar{x}_{0}) 
	\le&  \frac{\lambda_{0}^{'}}{3} \varrho^{T}[(\hat{c}+\hat{\rho})^{2}\otimes C^{T} C]\varrho + \frac{3}{\lambda_{0}^{'}}\bar{x}_{0}^{T}(G\mathcal{L}_{2}\mathcal{L}_{2}^{T}G \otimes C^{T}C) \bar{x}_{0}\\
	\le&  \frac{\lambda_{0}^{'}}{3} \varrho^{T}[(\hat{c}+\hat{\rho})^{2}\otimes C^{T} C]\varrho+ \frac{3\lambda_{max}(C^{T}C)\mathcal{L}_{2}^{T}GG\mathcal{L}_{2} }{\lambda_{0}^{'}\lambda_{min}(W)} \bar{x}_{0}^{T}W\bar{x}_{0},
	\end{aligned}
	\end{equation*}
	where $ \mathcal{L}_{2}^{T}GG\mathcal{L}_{2} $ is a scalar and $ \frac{W}{\lambda_{min}(W)} \ge I $ is used to arrive at the last inequality. Substituting above two inequalities into \eqref{v5_derivative}, we get
	\begin{equation*}
	\begin{aligned}
	\dot{V}_{3}
	\le & \varrho^{T} [G(\hat{c}+\hat{\rho}) \otimes(QA+A^{T}Q) -(\frac{\lambda_{0}^{'}}{3}(\hat{c}+\hat{\rho})^{2} -\frac{3G^{2}}{4\lambda_{0}^{'}} +\alpha G) \otimes C^{T}C]\varrho + (-\gamma + \frac{3\lambda_{max}(C^{T}C)\mathcal{L}_{2}^{T}GG\mathcal{L}_{2} }{\lambda_{0}^{'}\lambda_{min}(W)} ) \bar{x}_{0}^{T}W\bar{x}_{0}.
	\end{aligned}
	\end{equation*}
	Choosing $ \alpha \ge \frac{15 \lambda_{max}(G)}{4 \lambda_{0}^{'}}$ and $ \gamma = 1 + \frac{3\lambda_{max}(C^{T}C)\mathcal{L}_{2}^{T}GG\mathcal{L}_{2} }{\lambda_{0}^{'}\lambda_{min}(W)} $, we obtain
	\begin{equation}
	\begin{aligned}\label{v5_derivative_2}
	\dot{V}_{3}
	\le & -\varrho^{T} [G(\hat{c}+\hat{\rho}) \otimes W]\varrho - \bar{x}_{0}^{T}W\bar{x}_{0}\\
	\le &0
	\end{aligned}
	\end{equation}
	where the last inequality comes from $ W > 0 $. Similar as the proof in Theorem \ref{theorem_drected_notgood}, it is easy to verify that $ \varrho_{i}, \bar{x}_{0} $ and $ c_{i} $ are bounded, and the coupling weight $ c_{i}(t) $ converges to some finite value.
	
	Next we show the convergence of $ \zeta $ in \eqref{zeta}. Thanks to $ F=-Q^{-1} C^{T} $, it follows from LMI \eqref{lmi_stabilication} that
	$$
	(A+FC)^{T}Q+Q(A+FC)=A^{T}Q+Q A-2C^{T}C < 0.
	$$
	Therefore, $(A+FC)$ is Hurwitz and $ \zeta $ converges to zero.
	
	Then we try to verify the convergence of $ \psi $ in \eqref{psi_eta_tracking}. Based on $ \lim_{t \to \infty} \varrho =0, \lim_{t \to \infty} \bar{x}_{0}=0, \lim_{t \to \infty} \zeta =0 $ and $ (A+BK_{2}) $ being Hurwitz, from \eqref{psi_eta_tracking} we can conclude that $ \lim_{t \to \infty}\psi =0 $.
	
	Furthermore, based on $ \lim_{t \to \infty} \zeta =0, \lim_{t \to \infty}\psi =0 $, from \eqref{psi_eta_tracking} we can conclude that $ \lim_{t \to \infty}\eta =0 $.
	
	Finally, due to $ \hat{x}=\eta - \zeta $, based on $ \lim_{t \to \infty}\eta =0 $ and $ \lim_{t \to \infty}\zeta=0 $, we obtain $ \lim_{t \to \infty}\hat{x}=0 $. Recalling that $ \hat{x}=(\mathcal{L}_{1} \otimes I_{n})(\bar{x}- \textbf{1} \otimes x_{0}) $ and $ \mathcal{L}_{1} $ is a $ M $-matrix with all positive eigenvalues, we obtain $ \lim_{t \to \infty} (\bar{x}_{i}-x_{0})=\lim_{t \to \infty} (x_{i}-h_{i}-x_{0})=0 $, which means the distributed TVF tracking problem considering the leader of no input under directed spanning tree topology is solved where only a small subset of followers know leader's output information.
\end{proof}

\begin{remark}
	Compared with the TVF research \cite{dong_time-varying_2016-1}, where only the stabilization problem is solved, our control protocol in this section solves the TVF tracking problem and furthermore, is fully distributed due to the application of adaptive parameter $ c_{i}(t) $, while the protocol in \cite{dong_time-varying_2016-1} is not since its parameter depends on the smallest positive eigenvalue information of Laplacian matrices. The second improvement is that we use output measurements which are more applicable in reality than the state ones utilized in \cite{dong_time-varying_2016-1}. Thirdly, the protocol in \cite{dong_time-varying_2016-1} requires $ (A, B) $ to be controllable while we require $ (A, B) $ to be stabilizable, which is a more relaxed condition for system dynamics. Finally, the algorithm in \cite{dong_time-varying_2016-1} needs to check the TVF feasibility condition first, which is more complicated compared to our TVF shape information $ h(t) $ in \eqref{h_dynamic}. Furthermore, in contrast to the latest work~\cite{wang2017distributed} where the fully distributed TVF stabilization problem is solved with undirected communication topology among followers, our work is obviously a big improvement.
\end{remark}

Similar as Corollary \ref{corollary_stabilization} in Section \ref{sec_stabilizaiton} , for the special case where the relative state measurements $ x_{ij}= x_{i}-x_{j}, i, j \in \mathbb{F} $ are available among neighbors, the fully distributed adaptive relative state protocol is proposed for each follower $ i $ as
\begin{equation}\label{eq:protocol_directed_tracking_state}
\begin{aligned}
u_{i} =& K_{1}h_{i}+K_{2}v_{i},\\
\dot w_{i} =& Aw_{i} + Bu_{i} - BK_{1}h_{i} +B\tilde{F}[v_{i} -(x_{i}-h_{i})],\\
\dot v_{i} =& Av_{i}+ Bu_{i} - BK_{1}h_{i}  +B\tilde{F} (c_{i}+\rho_{i})(\psi_{i}-\eta_{i}) +B\tilde{F}[v_{i} -(x_{i}-h_{i})],\\
\dot c_{i} =& (\psi_{i}-\eta_{i})^{T}  \Gamma (\psi_{i}-\eta_{i}), i \in \mathbb{F}.
\end{aligned}
\end{equation}

\begin{corollary}\label{corollary_input}
	The fully distributed TVF tracking problem is solved with Assumptions \ref{assumptionABC} and \ref{assumptiondirected} under the protocol \eqref{eq:protocol_directed_tracking_state} if $ A+BK_{2} $ is Hurwitz, $ \Gamma = \tilde{Q}^{-1}BB^{T}\tilde{Q}^{-1} $, $ \tilde{F}=-B^{T}\tilde{Q}^{-1} $ and $ \rho_{i}=\varrho_{i}^{T}\tilde{Q}^{-1}\varrho_{i}$, where $ \tilde{Q} > 0 $ satisfies the LMI \eqref{lmi_stabilication_state}.
	And $ c_{i}(t), i \in \mathbb{F} $ converge to some finite steady-state values.
\end{corollary}
The proof is similar as the details in Theorem~\ref{theorem_tracking} and is omitted for conciseness.

\subsection{Directed formation tracking with bounded leader input}\label{section_bounded_input}
In the previous sections, we dealt with TVF tracking control problem without leader's input for general linear MAS. In this section, we extend our analysis to address formation tracking issue with leader's control input $ u_{0}(t) $.
\begin{assumption}\label{leader input}
	The leader's control input satisfies that $ \|u_{0}(t)\| \le \epsilon $, where $ \epsilon $ is a positive constant.
\end{assumption}

Based on the protocol \eqref{eq:protocol_directed_tracking} in Section \ref{section_tracking}, the following fully distributed adaptive protocol is proposed to solve TVF tracking problem with leader's bounded input as
\begin{equation}\label{eq:protocol_tracking_input}
\begin{aligned}
u_{i} =& K_{1}h_{i}+K_{2}v_{i}- \beta z(B^{T}S\eta_{i}), \\
\dot w_{i} =& Aw_{i} + Bu_{i} - BK_{1}h_{i} +F[Cw_{i} -(y_{i}-Ch_{i})],\\
\dot v_{i} =& Av_{i}+ B[u_{i}-\beta z(B^{T}Q(\psi_{i} - \eta_{i}))] - BK_{1}h_{i}  + FC (c_{i}+\rho_{i}) (\psi_{i}-\eta_{i}) +FC(w_{i} - \bar{x}_{i}),\\
\dot c_{i} =& (\psi_{i}-\eta_{i})^{T} C^{T} \Gamma C (\psi_{i}-\eta_{i}), i \in \mathbb{F}
\end{aligned}
\end{equation}
where $ c_{i}(0)\ge 1 $, $ S \ge 0 $, $ \psi_{i}= \sum_{j=0}^{N} a_{ij}(v_{i}-v_{j}), \eta_{i} = \sum_{j=0}^{N} a_{ij}(w_{i}-w_{j}), i \in \mathbb{F} $ and $ w_{0}=Aw_{0} +Bu_{0}+F(Cw_{0}-y_{0}), v_{0}=0 $. The positive constant $ \beta $ is to be determined later and other parameters are the same as in \eqref{eq:protocol_directed_tracking} of Section \ref{section_tracking}. The nonlinear function $ z(\cdot) $ is defined as
\begin{equation}\label{zx_discontinuous}
z(x) =
\begin{cases}
\frac{x}{\|x\|}  & \text{if $ \|x\| \neq 0$,} \\
0 & \text{if $ \|x\| = 0$.}
\end{cases}
\end{equation}
Similar as in Section \ref{section_tracking}, combine \eqref{eq:protocol_tracking_input} with \eqref{eq:dynamics} then
\begin{equation}\label{psi_eta_tracking_input}
\begin{aligned}
\dot{\hat{x}}=& (I_{N} \otimes A)\hat{x} +(I_{N} \otimes BK_{2})\psi - (\mathcal{L}_{1} \otimes B)(\beta M(\eta) + \textbf{1}\otimes u_{0}),\\
\dot{\eta} =& (I_{N} \otimes A)\eta +(I_{N} \otimes BK_{2})\psi + (I_{N} \otimes FC)(\eta - \hat{x}) - (\mathcal{L}_{1} \otimes B)(\beta M(\eta) + \textbf{1}\otimes u_{0}),\\
\dot{\psi}=& [I_{N} \otimes (A+BK_{2})] \psi + [\mathcal{L}_{1}(\hat{c}+\hat{\rho}) \otimes FC]\varrho + (I_{N} \otimes FC)(\eta - \hat{x}) +(\mathcal{L}_{1} \otimes FC)[\textbf{1} \otimes (w_{0}-x_{0})]\\
&- (\mathcal{L}_{1} \otimes B)\beta [M(\eta) + Z(\varrho)], i \in \mathbb{F}\\
\end{aligned}
\end{equation}
where $ Z(\varrho)=[z(B^{T}Q(\psi_{1} - \eta_{1}))^{T}, \ldots , z(B^{T}Q(\psi_{N} - \eta_{N}))^{T}]^{T}, M(\eta)=[z(B^{T}S\eta_{1})^{T}, \ldots, z(B^{T}S\eta_{N})^{T}]^{T} $, and
\begin{equation}\label{zeta_input}
\begin{aligned}
\dot{\zeta}=& [I_{N} \otimes (A+FC)]\zeta,\\
\dot{\varrho} =&  [I_{N} \otimes A + \mathcal{L}_{1}(\hat{c}+\hat{\rho}) \otimes FC]\varrho + (\mathcal{L}_{1}\otimes FC)(\textbf{1}\otimes \bar{x}_{0})- (\mathcal{L}_{1} \otimes B)(\beta Z(\varrho) - \textbf{1}\otimes u_{0}),\\
\dot c_{i} =& \varrho_{i}^{T} C^{T} \Gamma C \varrho_{i}, i \in \mathbb{F}.
\end{aligned}
\end{equation}

The following theorem presents a result of designing protocol \eqref{eq:protocol_tracking_input} to solve the TVF tracking problem with leader's bounded input under directed topology.
\begin{theorem}\label{theorem_tracking_input}
	Suppose Assumptions \ref{assumptionABC}, \ref{assumptiondirected} and \ref{leader input} hold, the fully distributed TVF tracking problem with leader's bounded input is solved under the protocol \eqref{eq:protocol_tracking_input} if $ A+BK_{2} $ is Hurwitz, $ \Gamma = I $, $ F=-Q^{-1} C^{T} $ and $ \rho_{i}=\varrho_{i}^{T}Q\varrho_{i}$, where $ Q > 0 $ is a solution to the LMI \eqref{lmi_stabilication}. $ \beta \ge \epsilon$ and $ S > 0 $ satisfies
	\begin{equation}\label{S}
	S(A+BK_{2}) +(A+BK_{2})^{T}S <0.
	\end{equation}
	Moreover, the coupling weight $ c_{i}(t), i \in \mathbb{F} $ converge to some finite steady-state values.
\end{theorem}

\begin{proof}
	First, based on the proof of Theorem \ref{theorem_tracking}, the convergence of $ \zeta $ in \eqref{zeta_input} is addressed. Then in order to prove $ \lim_{t \to \infty}\varrho =0 $, let
	\begin{equation}\label{v6}
	V_{4}=\frac{1}{2}\sum_{i=1}^{N} g_{i}(2c_{i}+\rho_{i})\rho_{i}+\frac{1}{2}\sum_{i=1}^{N} g_{i}(c_{i}- \alpha)^{2} + \gamma \bar{x}_{0}^{T}Q\bar{x}_{0}.
	\end{equation}
	By choosing the same parameters $ \alpha $ and $ \gamma $ as in the proof of Theorem \ref{theorem_tracking}, we get
	\begin{equation}
	\begin{aligned}\label{v6_derivative_input}
	\dot{V}_{4} 
	\le & -\varrho^{T} [G(\hat{c}+\hat{\rho}) \otimes W]\varrho - \bar{x}_{0}^{T}W\bar{x}_{0}-2\varrho^{T}[G(\hat{c}+\hat{\rho})\mathcal{L}_{1}\otimes QB ](\beta Z(\varrho) - \textbf{1}\otimes u_{0}).
	\end{aligned}
	\end{equation}
	Note that
	\begin{equation*}
	\begin{aligned}
	\varrho_{i}^{T}QBz(B^{T}Q\varrho_{i})=&\varrho_{i}^{T}QB \frac{B^{T}Q\varrho_{i}}{\|B^{T}Q\varrho_{i}\|}=\|B^{T}Q\varrho_{i}\|, \\
	\varrho_{i}^{T}QBz(B^{T}Q\varrho_{j})\le&\|\varrho_{i}^{T}QB\| \left\| \frac{B^{T}Q\varrho_{j}}{\|B^{T}Q\varrho_{j}\|} \right\| =\|B^{T}Q\varrho_{i}\|,
	\end{aligned}
	\end{equation*}
	then
	\begin{equation}\label{z1}
	\begin{aligned}
	-\varrho^{T}[G(\hat{c}+\hat{\rho})\mathcal{L}_{1}\otimes QB ]\beta Z(\varrho)
	=& -\sum_{i=1}^{N}g_{i}(c_{i}+\rho_{i})  \beta  \varrho_{i}^{T} QB \left[  \sum_{j=1}^{N}a_{ij}(z(B^{T}Q\varrho_{i})-z(B^{T}Q\varrho_{j}) ) +a_{i0}z(B^{T}Q\varrho_{i}) \right]  \\
	\le & -\sum_{i=1}^{N}g_{i}(c_{i}+\rho_{i})  \|B^{T}Q\varrho_{i}\| a_{i0} \beta.
	\end{aligned}
	\end{equation}
	On the other hand, using $ \mathcal{L}_{1}\textbf{1}=-\mathcal{L}_{2} $, we get
	\begin{equation}\label{z2}
	\begin{aligned}
	\varrho^{T}[G(\hat{c}+\hat{\rho})\mathcal{L}_{1}\otimes QB] (\textbf{1} \otimes u_{0})
	=& \sum_{i=1}^{N}g_{i}(c_{i}+\rho_{i}) \varrho_{i}^{T} QBa_{i0} u_{0}
	\le  \sum_{i=1}^{N}g_{i}(c_{i}+\rho_{i}) \|B^{T}Q\varrho_{i}\| a_{i0} \epsilon.
	\end{aligned}
	\end{equation}
	Substitute \eqref{z1} and \eqref{z2} into \eqref{v6_derivative_input} with $ \beta \ge \epsilon $, then
	\begin{equation}\label{v6_final}
	\begin{aligned}
	\dot{V}_{4} 
	\le & -\varrho^{T} [G(\hat{c}+\hat{\rho}) \otimes W]\varrho - \bar{x}_{0}^{T}W\bar{x}_{0}\\
	=& - \xi ^{T} (I_{N+1} \otimes W ) \xi 
	\le 0
	\end{aligned}
	\end{equation}
	where the last inequality comes from $ W=-(Q A+A^{T}Q-2C^{T}C) >0 $ and $ \xi=[\varrho^{T}(\sqrt{G(\hat{c}+\hat{\rho}) \otimes I_{n}}), \bar{x}_{0}^{T}]^{T} $. Similar as the proof in Theorem \ref{theorem_tracking}, it is easy to verify that $ V_{4}, \varrho_{i}, \bar{x}_{0}, c_{i} $ are bounded and $ c_{i}(t) $ converges to some finite value.
	
	By the definition of $\xi$ and the bounded property of $ \varrho_{i}, \bar{x}_{0} $, it can be get that  $\xi$ is bounded. In addition, since  $u_{0}(t)$ is bounded in Assumption \ref{leader input}, $\dot{\varrho}$ in~\eqref{zeta_input} is bounded, too. Recall that $ \dot{\bar{x}}_{0}=\dot{w}_{0}-\dot{x}_{0}=(A+FC)\bar{x}_{0} $ is also bounded, which furthermore implies that $\dot{\xi }$ is bounded.
	
	Integrate \eqref{v6_final} then
	$$
	\int _{0}^{\infty} \xi ^{T} (I_{N+1} \otimes W ) \xi dt \le V_{4}(0)- V_{4}(\infty).
	$$
	
	Since $V_{4}(\infty)$ is finite due to $\dot{V}_{4} \le 0$ and $ V_{4}(t)>0 $, we get that $\int _{0}^{\infty} \xi ^{T} (I_{N+1} \otimes W ) \xi dt$ has a finite limit.
	
	In fact, $2\xi ^{T} (I_{N+1} \otimes W ) \dot \xi $ is bounded because of the boundedness of $\xi$ and $\dot{\xi }$, which in turn proves that $\xi ^{T} (I_{N+1} \otimes W ) \xi$ is uniformly continuous.
	
	Finally, $\int _{0}^{\infty} \xi ^{T} (I_{N+1} \otimes W ) \xi dt$ is differentiable and has a finite limit as $t \to \infty$, and $\xi ^{T} (I_{N+1} \otimes W ) \xi$ is uniformly continuous. Then by Barbalat's Lemma \cite{khalil1996noninear} we get
	$\xi ^{T} (I_{N+1} \otimes W ) \xi \to 0$ as $t \to \infty$, which means $ \lim_{t \to \infty}\xi = 0 $ such that $  \lim_{t \to \infty}\varrho = 0$.
	
	Next, to prove $ \lim_{t \to \infty}\eta = 0 $, we consider the following Lyapunov function candidate
	\begin{equation}\label{v7}
	V_{5}=\eta^{T} (I_{N}\otimes S)\eta + \gamma_{1} \zeta^{T}(I_{N}\otimes Q)\zeta + \gamma_{2} V_{4}
	\end{equation}
	where $ \gamma_{1}, \gamma_{2} $ are positive constants to be determined later. $V_{5} $ is positive definite with respect to $ \eta, \zeta, \varrho, c_{i} $ and $ \bar{x}_{0} $. Combining \eqref{psi_eta_tracking_input} and \eqref{zeta_input}, the derivative of $ V_{5} $ is
	\begin{equation}\label{v7_derivative}
	\begin{aligned}
	\dot{V}_{5} 
	= & -\eta^{T} (I_{N}\otimes \bar{W})\eta -\gamma_{1} \zeta^{T}(I_{N}\otimes W)\zeta+2\eta^{T}(I_{N}\otimes SBK_{2})\varrho +2\eta^{T}(I_{N}\otimes SFC)\zeta\\
	&- 2\eta^{T}(L_{1}\otimes SB)(\beta M(\eta)+ \textbf{1}\otimes u_{0}) +\gamma_{2}\dot{V}_{4}
	\end{aligned}
	\end{equation}
	where $ \bar{W}=-[S(A+BK_{2}) +(A+BK_{2})^{T}S] > 0 $ and $ W=-(Q A+A^{T}Q-2C^{T}C) >0 $. Similarly by using Lemma \ref{lemma_pq}, we have
	\begin{equation}\label{v7_1}
	\begin{aligned}
	2\eta^{T}(I_{N}\otimes SBK_{2})\varrho \le & \frac{1}{4}\eta^{T}(I_{N}\otimes \bar{W})\eta + \frac{4\lambda_{max}(K_{2}^{T}B^{T}SSBK_{2})}{\lambda_{min}(\bar{W})} \varrho^{T}\varrho,\\
	2\eta^{T}(I_{N}\otimes SFC)\zeta \le & \frac{1}{4}\eta^{T}(I_{N}\otimes \bar{W})\eta + \frac{4\lambda_{max}(K_{2}^{T}B^{T}SSBK_{2})}{\lambda_{min}(\bar{W})} \zeta^{T}\zeta.
	\end{aligned}
	\end{equation}
	Due to $ c_{i}(0)\ge 1$, $ \dot{c}_{i} \ge 0$ and $ \rho_{i}=\varrho_{i}^{T}Q\varrho_{i} \ge 0 $
	we get $ (\hat{c}+\hat{\rho}) > I $. Choosing $ \gamma_{1} \ge  \frac{4\lambda_{max}(C^{T}F^{T}SSFC)}{\lambda_{min}(\bar{W}) \lambda_{min}(W) }, \gamma_{2} \ge \frac{4\lambda_{max}(K_{2}^{T}B^{T}SSBK_{2})}{\lambda_{min}(\bar{W}) \lambda_{min}(W) \lambda_{min}(G) } $ and substituting \eqref{v6_final}, \eqref{v7_1} into \eqref{v7_derivative}, we obtain
	\begin{equation}
	\begin{aligned}
	\dot{V}_{5} 
	= & -\frac{1}{2}\eta^{T}(I_{N}\otimes \bar{W})\eta- 2\eta^{T}(\mathcal{L}_{1}\otimes SB)(\beta M(\eta)+ \textbf{1}\otimes u_{0}).
	\end{aligned}
	\end{equation}
	Similar as in \eqref{z1} and \eqref{z2}, we can prove that
	$$
	- 2\eta^{T}(\mathcal{L}_{1}\otimes SB)(\beta M(\eta)+ \textbf{1}\otimes u_{0}) \le 0.
	$$
	Finally,
	\begin{equation*}
	\begin{aligned}
	\dot{V}_{5} 
	= & -\frac{1}{2}\eta^{T}(I_{N}\otimes \bar{W})\eta \le 0.
	\end{aligned}
	\end{equation*}
	Therefore, $ V_{5} $ is bounded and so is $ \eta $. It is easy to verify that $ \lim_{t \to \infty} \eta =0 $. Due to $ \lim_{t \to \infty}\zeta=0 $ and $ \hat{x}= \eta-\zeta $, we get $ \lim_{t \to \infty}\hat{x}=0 $. Similar as the proof of Theorem~\ref{theorem_tracking},
	the distributed adaptive TVF tracking problem considering the leader's bounded input with directed spanning tree topology is solved.
\end{proof}

\begin{remark}
	Compared to the previous protocols without leader's input, the nonlinear components $ z(B^{T}S\eta_{i}) $ and $ z(B^{T}Q\varrho_{i}) $ in protocol \eqref{eq:protocol_tracking_input} are used to deal with the leader's bounded input.
	It is worth noting that the techniques utilized in the proof are partially motivated by \cite{lv_fully_2015} where the distributed output feedback consensus problem for general linear MAS has been studied by using a sequential observer design approach.
\end{remark}
\begin{remark}
	Since function~\eqref{zx_discontinuous} is nonsmooth, the whole control protocol~\eqref{eq:protocol_tracking_input} is discontinuous dealing with the leader's bounded input $ u_{0}(t) $. In fact, from Subsection~\ref{section_tracking} to \ref{section_bounded_input}, we regard $ u_{0}(t) $ as one kind of disturbances and use function~\eqref{zx_discontinuous} to compensate it. The discontinuous protocol~\eqref{eq:protocol_tracking_input} can be modified to be continuous with the following smooth function $ z(x) $
	\begin{equation}
	z(x) =
	\begin{cases}
	\frac{x}{\|x\|}  & \text{if $ \|x\| > \delta$,} \\
	\frac{x}{\delta} & \text{if $ \|x\| \le \delta$}
	\end{cases}
	\end{equation}
	and $ \dot c_{i} = (\psi_{i}-\eta_{i})^{T} C^{T} \Gamma C (\psi_{i}-\eta_{i}) -\varepsilon_{i}(c_{i}(t)-1), i \in \mathbb{F}, $ where $ \varepsilon_{i}, i \in \mathbb{F} $ and $ \delta $ are small positive constants. It is worth noting that this modified continuous protocol's control effect will be uniformly ultimately bounded while protocol~\eqref{eq:protocol_tracking_input} make the TVF tracking error converge to zero asymptotically. Since this paper focus on proposing the unified framework of TVF control design from undirected to directed topology, from stabilization to tracking and from a leader without input to a one with bounded input $ u_{0}(t) $, we will not go into the proving detail about the modified protocol.
\end{remark}

\section{Simulation}\label{simulaiton}

\begin{figure}
	\centering
	\begin{subfigure}[h]{0.45\textwidth}
		\includegraphics[width=\textwidth]{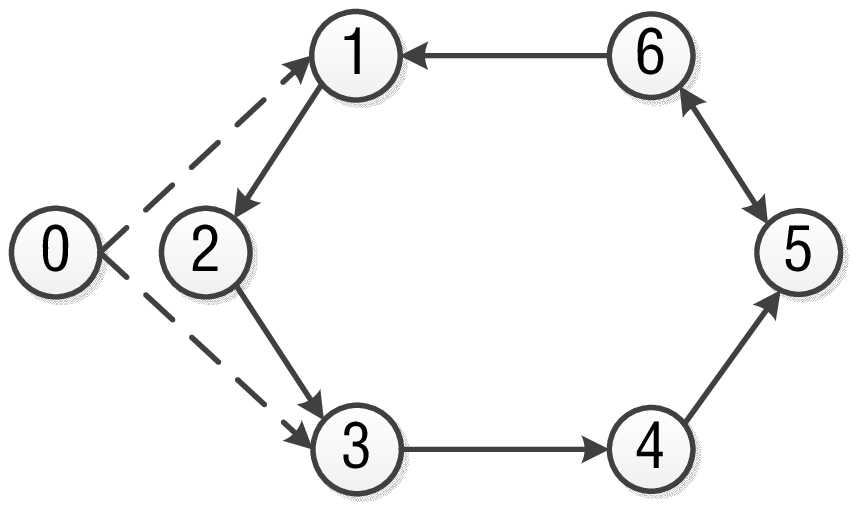}
		\caption{Among the multi-agent system.}
		\label{communication topology}
	\end{subfigure}\quad
	~ 
	\begin{subfigure}[h]{0.45\textwidth}
		\includegraphics[width=0.8\hsize]{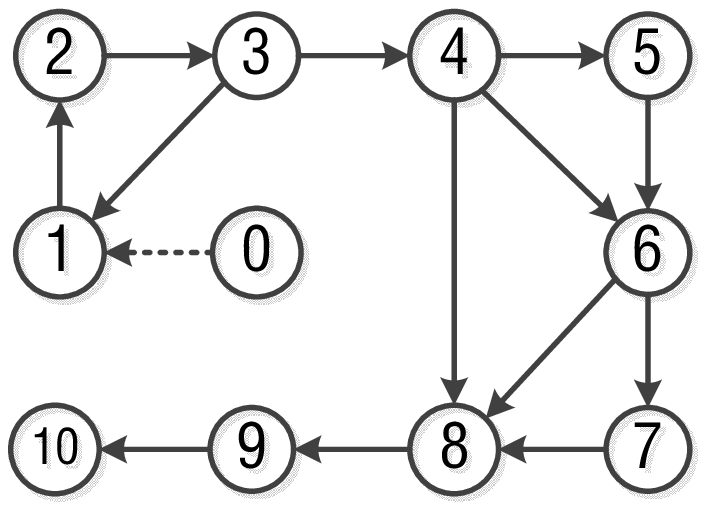}
		\caption{ Among multiple nonholonomic vehicles.}
		\label{communication topology_nonholo}
	\end{subfigure} 
	\caption{The communication topology $ \mathcal{G} $.}\label{fig:communication}
\end{figure}

The effectiveness of the proposed TVF control laws given in Section~\ref{mainresults} are demonstrated in this section by two numerical examples. The first example is presented to illustrate that the fully distributed adaptive controller~\eqref{eq:protocol_tracking_input} successfully achieve the TVF tracking with the leader of bounded input. Then, an application to nonholonomic mobile robots with the controller~\eqref{eq:protocol_tracking_input} in Section~\ref{section_bounded_input} is provided in the second example.

\textbf{Example 1.}
Consider a group of agents consisting of a leader labeled 0 and six followers labeled from 1 to 6, where the communication topology $ \mathcal{G} $ is shown in Fig. \ref{communication topology} satisfying Assumption~\ref{assumptiondirected}. Suppose that the state of agent $ i $ in~\eqref{eq:dynamics} is described as $ x_{i}(t) = (x_{i1}(t), \ldots, x_{i6}(t))^{T} \in \mathbb{R}^{6} $. $ A $ and $B $ in~\eqref{eq:dynamics} are given as follows
$$
A =  
\begin{bmatrix}
0 & 0 & 0 & 1 & 0& 0\\
0 & 0 & 0 & 0& 1& 0\\
0 & 0 & 0 & 0& 0& 1\\
-1 & 0 & 0 & 0& 0& 0\\
0 & -1 & 0 & 0& 0& 0\\
0 & 0& -1 & 0 & 0& 0
\end{bmatrix}, \quad 
B =\begin{bmatrix}
0 & 0 & 0\\
0 & 0 & 0\\
0 & 0 & 0\\
1 & 0& 0\\
0 & 1& 0\\
0 & 0& 1
\end{bmatrix}.
$$
Choosing $ C=[I_{5}, 0_{5 \times 1}] $, then it is easy to verify that $ (A,B,C) $ is stabilizable and detectable. 

\begin{figure}
	\centering
	\begin{subfigure}[h]{0.45\textwidth}
		\includegraphics[width=\textwidth]{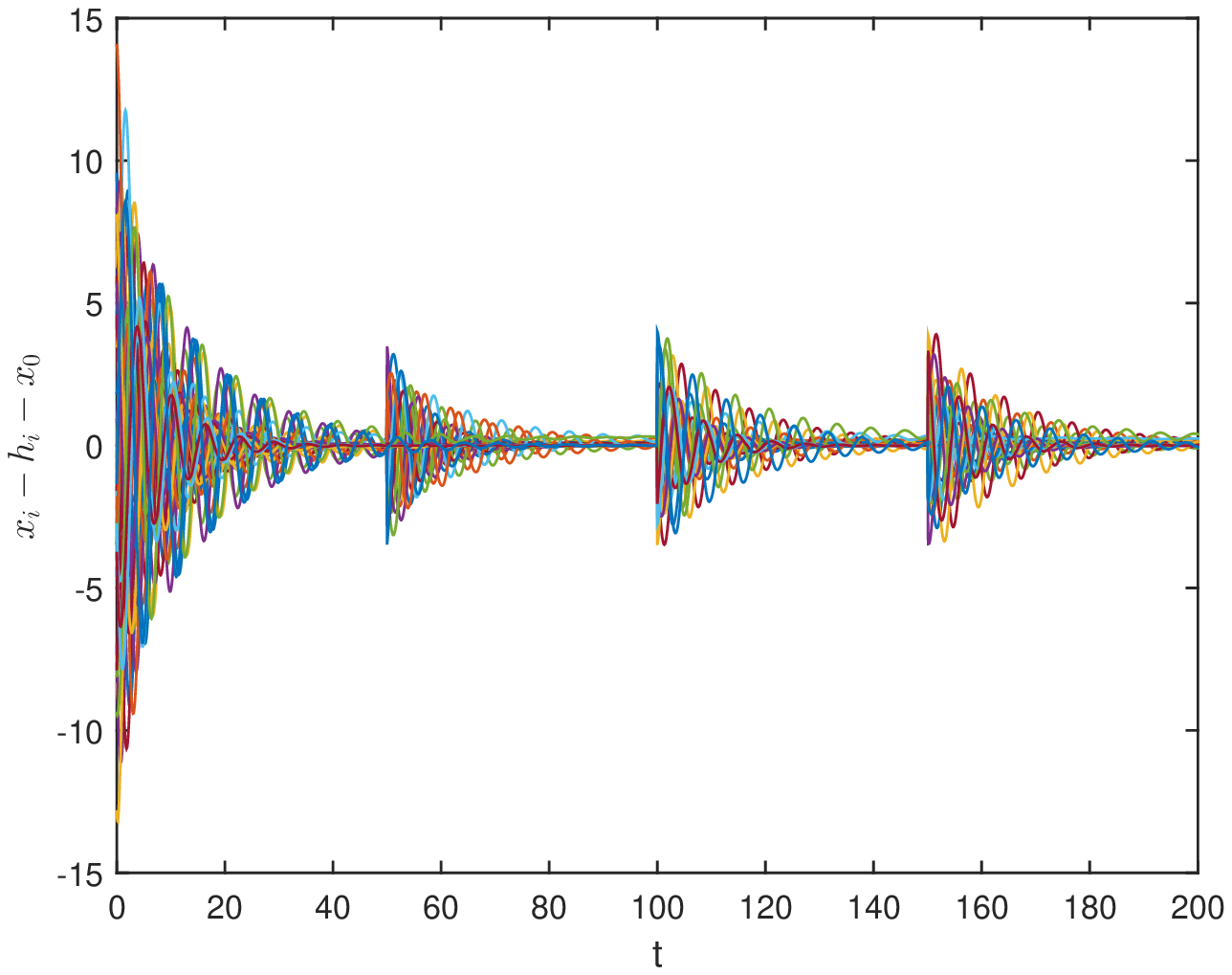}
		\caption{ $ \tilde{x}_{i}=x_{i}-h_{i}-x_{0} $.}
		\label{formation error}
	\end{subfigure}\quad
	~ 
	\begin{subfigure}[h]{0.45\textwidth}
		\includegraphics[width=\textwidth]{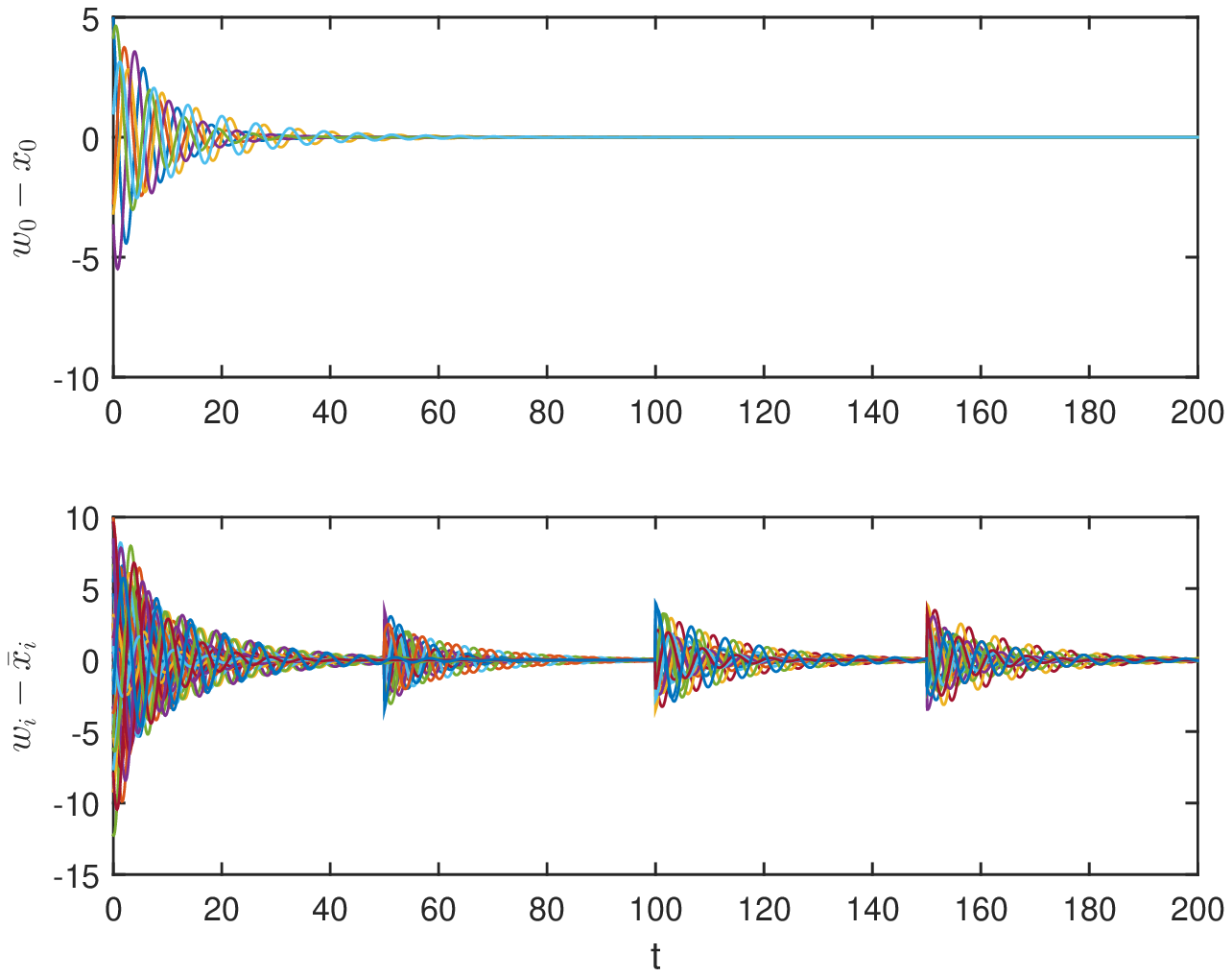}
		\caption{ $ \bar{x}_{0}=w_{0}-x_{0} $
			(top) and $ w_{i} - \bar{x}_{i} $ (bottom).}
		\label{estimation error}
	\end{subfigure} 
	\caption{The control errors.}\label{fig:control_errors}
\end{figure}

Inspired by our previous work~\cite{jiang2017distributed}, definite $ h_{i1}=-r \, cos(wt+(i-1) \pi /3) +r \, sin(wt+(i-1) \pi /3),
h_{i2}= 2r \, sin(wt+(i-1) \pi /3),
h_{i3}=2r \, cos(wt+(i-1)\pi/3),
h_{i4}=w \, r \, cos(wt+(i-1) \pi /3) +w \, r \, sin(wt+(i-1) \pi /3),
h_{i5}= 2w \, r \, cos(wt+(i-1) \pi /3),
h_{i6}= -2w \, r \, sin(wt+(i-1)\pi/3), \, i=1, \ldots, 6 $,
where $ r=2, w=2 $ for followers and $ h_{i}=[h_{i1},\ldots,h_{i6}]^{T} $. The TVF shapes for followers are described as the parallel hexagon shape when $ t \in [0,50) \cup [150, 200] $, the parallelogram shape when $ t \in [50,100) $ and the triangle shape when $ t \in [100,150) $ in the following

	\begin{equation*}
	h(t)  = \begin{cases}
	[h_{1}^{T}, h_{2}^{T}, h_{3}^{T}, h_{4}^{T}, h_{5}^{T}, h_{6}^{T}]^{T}  & \text{$ 0 \leq t < 50 $,} \\
	[h_{1}^{T}, (\frac{h_{1}+h_{3}}{2})^{T}, h_{3}^{T}, h_{4}^{T}, (\frac{h_{4}+h_{6}}{2})^{T}, h_{6}^{T}]^{T}  & \text{$ 50 \leq t < 100$,} \\
	[h_{1}^{T}, (\frac{h_{1}+h_{3}}{2})^{T}, h_{3}^{T}, (\frac{h_{3}+h_{5}}{2})^{T}, h_{5}^{T}, (\frac{h_{5}+h_{1}}{2})^{T}]^{T}  & \text{$ 100 \leq t < 150 $,}\\
	[h_{1}^{T}, h_{2}^{T}, h_{3}^{T}, h_{4}^{T}, h_{5}^{T}, h_{6}^{T}]^{T}  & \text{$ 150 \leq t \leq 200$.}
	\end{cases}
	\end{equation*}
It is obvious that $ \lim_{t \to \infty} \sum_{i=1}^{6} h_{i}(t)=0 $ , meaning that the six followers will keep TVF shapes around the leader when the desired formation tracking is achieved. 

Define the leader's bounded input as $ u_{0}(t)=[e^{-t}+1, e^{-2t}, 2+sin(\frac{t}{2})]^{T} $ and $ \beta =4 $. From \eqref{h_dynamic} we get
$$ K_{1}=\begin{bmatrix}
-3 & 0 & 0 & 0 & 0 & 0\\
0 & -3 & 0 & 0& 0 & 0\\
0 &0 & -3 & 0 & 0& 0
\end{bmatrix}. $$
Solving LMI \eqref{lmi_stabilication} gives a solution\\
\small{$$ Q=\begin{bmatrix}
	7.314  & -0.000&   -0.000&   -0.000&    0.000&    0.000\\
	-0.000  &  7.314 &  -0.000 &  -0.000 &   0.000 &  -0.000\\
	-0.000  & -0.000  &  7.412   &-0.000 &   0.000   &-0.487\\
	-0.000&   -0.000&   -0.000    &7.314&   -0.000 &  -0.000\\
	0.000    &0.000    &0.000&   -0.000    &7.314 &  -0.000\\
	0.000   &-0.000   &-0.487 &  -0.000  & -0.000 &   7.412
	\end{bmatrix}. $$}
Thus the feedback gain matrix in \eqref{eq:protocol_tracking_input} is obtained as\\
$$ F=\begin{bmatrix}
-0.1367&   -0.0000&   -0.0000&   -0.0000 &   0.0000\\
-0.0000  & -0.1367&   -0.0000  & -0.0000   & 0.0000\\
-0.0000&   -0.0000&   -0.1355&   -0.0000 &   0.0000\\
-0.0000  & -0.0000&   -0.0000&   -0.1367 &  -0.0000\\
0.0000&    0.0000&    0.0000&   -0.0000&   -0.1367\\
0.0000  & -0.0000 &  -0.0089 &  -0.0000 &  -0.0000
\end{bmatrix}. $$
Assign eigenvalue of $ A+BK_{2} $ as $ -1,-5,-10+10j, -10-10j, -20, -50 $ with $ j^{2}=-1 $, then 
$$ K_{2}=\begin{bmatrix}
-111.3&   21.8 &  10.7&  -13.8&  -11.3&  -11.0\\
56.3& -159.0 & -10.9 &  -1.8 & -25.5 &  -2.5\\
107.6 & -78.1 & -71.1 &  29.2 & -32.4&  -56.7
\end{bmatrix}. $$
\begin{figure}
	\centering
	\begin{subfigure}[h]{0.45\textwidth}
		\includegraphics[width=\textwidth]{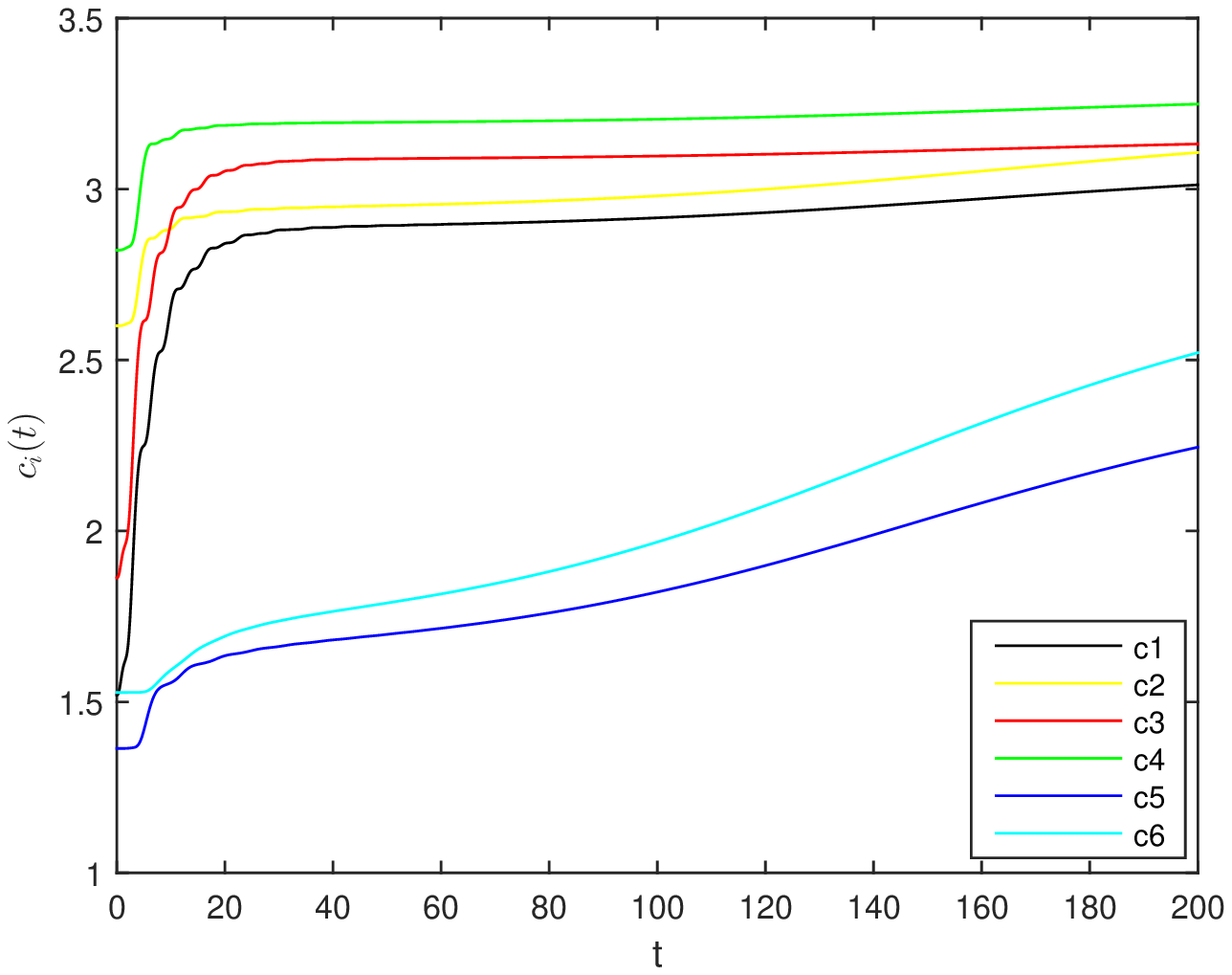}
		\caption{ Among the multi-agent system.}
		\label{ci}
	\end{subfigure}\quad
	~ 
	\begin{subfigure}[h]{0.45\textwidth}
		\includegraphics[width=\textwidth]{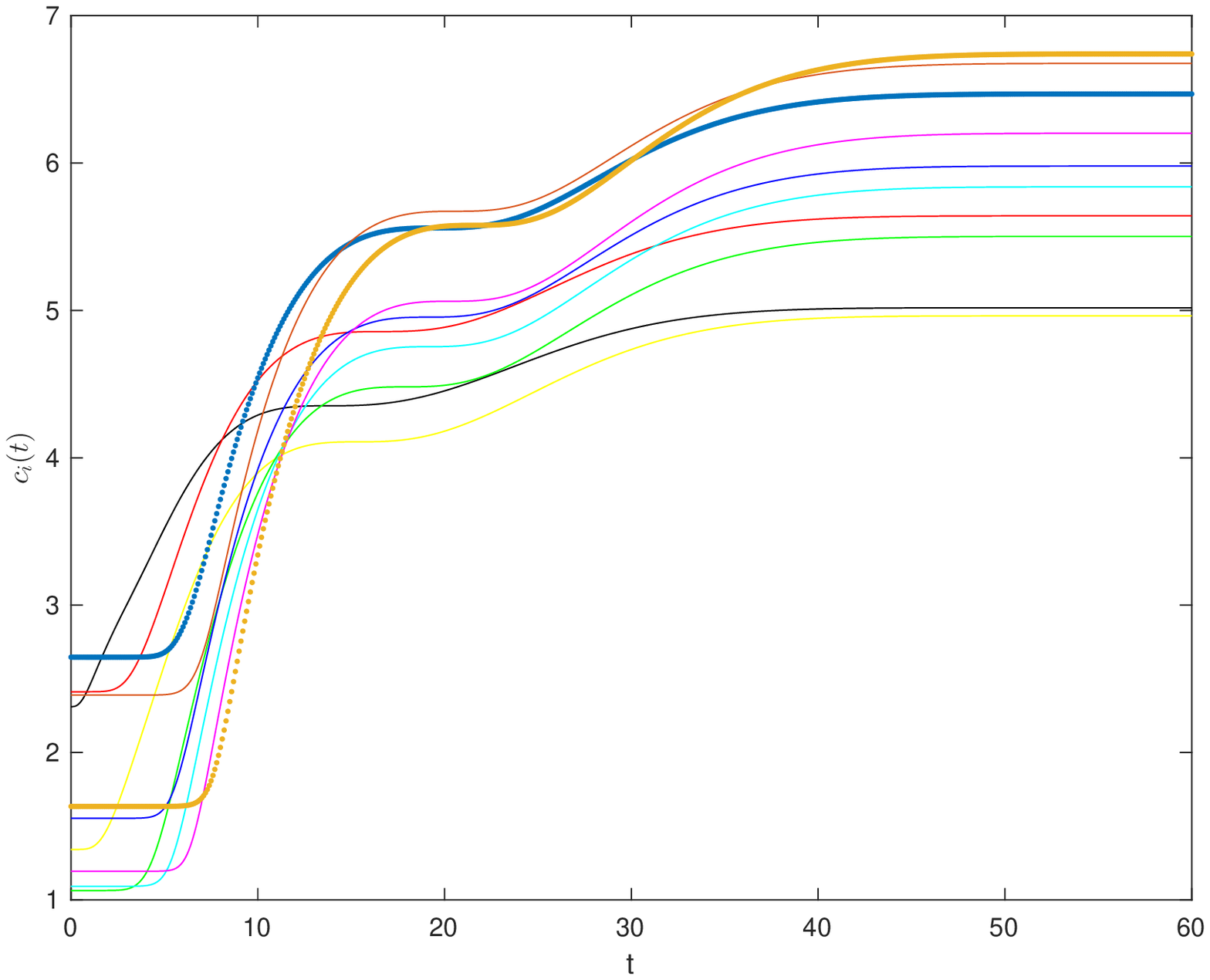}
		\caption{ Among multiple nonholonomic vehicles.}
		\label{ci_nonholo}
	\end{subfigure} 
	\caption{The coupling weights $ c_{i}(t) $.}\label{fig:ci}
\end{figure}
Substituting $ K_{2} $ into \eqref{S}, we get
$$ S=\begin{bmatrix}
2319.3 &  -422.9 &  -383.7 &   21.1  & -7.4  &  21.7\\
-422.9 &   2453.7  &  186.0 &  -11.3 &   9.3 &  -16.2\\
-383.7  &  186.0 &   1167.9 &  -6.1  &  3.1 &   3.9\\
21.1  & -11.3 &  -6.1  &  24.6&   -0.7 &  -0.2\\
-7.4  &  9.3 &   3.1  & -0.7  &  19.5 &  -1.0\\
21.7  & -16.2 &   3.9 &  -0.2 &  -1.0 &   16.9
\end{bmatrix}. $$
\begin{figure*}[!htb]
	\centering
	\includegraphics[width=\hsize]{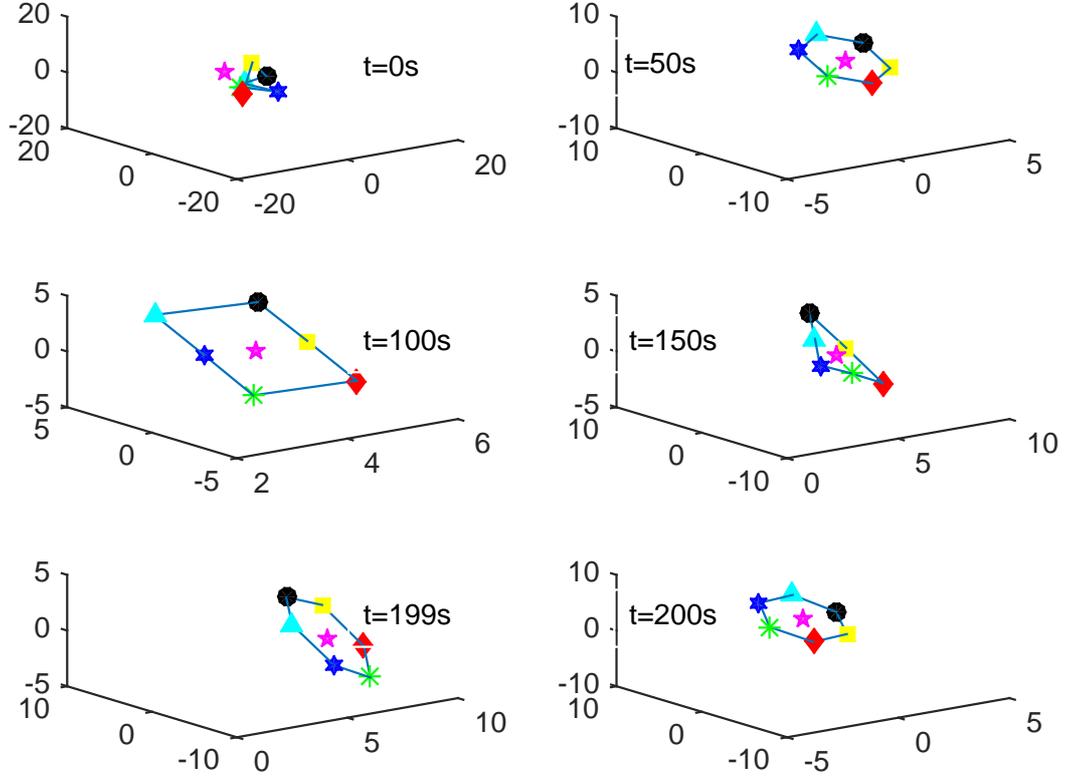}
	\caption{Position snapshots of six followers (circle, square, diamond, asterisk, hexagon, triangle) and the leader (pentagram) forming the shape from parallel hexagon to parallelogram, then triangle and finally back to parallel hexagon.}
	\label{snapshots}
\end{figure*}
Set the initial states $ x_{ij}(0)= 10\delta -5, c_{i}(0)=2\delta +1, i,j=1, \ldots, 6 $ for followers and $ x_{0j}(0)= 10\delta -5, j=1, \ldots, 6 $ for the leader, where $ \delta $ is a pseudorandom value with a uniform distribution on the interval $ (0,1) $.
Fig. \ref{formation error} shows the TVF tracking error $ x_{i}-h_{i}-x_{0} $, meaning that the time-varying output formation problem is indeed solved. 
Fig. \ref{estimation error} describes the leader's state estimation error $ \bar{x}_{0}=w_{0}-x_{0} $
(top) and followers' formation stabilization estimation error $ w_{i} - \bar{x}_{i} $, respectively, which means the distributed observers are designed correctly.
The coupling weights shown in Fig. \ref{ci} converge to some finite values. 
Fig. \ref{snapshots} depicts the position snapshots of followers and the leader at different timestamps. 
Six followers form formation shapes with random initial positions and keep tracking the leader which is located at the shape center at the same time. The TVF shapes change from parallel hexagon to parallelogram, then triangle and finally back to parallel hexagon.
From $ t=199s $ and $ t=200s $, we can see the shape keeps rotating around the leader, which means it is time-varying. The leader's trajectory is time-varying as well. It is worth noting that the presented results can be applied to target enclosing problems with regarding the leader as the target.

\textbf{Example 2.}
The MAS have many applications in reality. For instance, to accomplish an unknown area detection task, a group of autonomous nonholonomic vehicles is a good choice. Each vehicle can be equipped with different sensors. The multi-vehicle system should form a formation shape, and rotates the shape at the same time so that each direction is detected by different sensors.

Consider a group of eleven mobile vehicles with communication topology shown in Fig.~\ref{communication topology_nonholo}, each of which has the motion equations as follows
\begin{equation}
\begin{pmatrix}
\dot{r}_{xi}\\\dot{r}_{yi}\\\dot{\theta}_{i}\\\dot{\bar{v}}_{i}\\\dot{\bar{w}}_{i}
\end{pmatrix}=\begin{pmatrix}
\bar{v}_{i}cos(\theta_{i})\\\bar{v}_{i}sin(\theta_{i})\\\bar{w}_{i}\\0\\0
\end{pmatrix}+
\begin{pmatrix}
0&0\\0&0\\0&0\\\frac{1}{m_{i}} &0\\0 & \frac{1}{J_{i}}
\end{pmatrix} \begin{pmatrix}
F_{i}\\ \tau_{i}
\end{pmatrix}
\end{equation}
where$ (r_{xi}, r_{yi}) $ is the Cartesian position of the $ i $-th vehicle, $ \theta_{i} $ is the orientation, and $ \bar{v}_{i}, \bar{w}_{i} $ are the linear speed and angular speed, respectively. $ m_{i} $ is the mass, $ J_{i} $ is the moment of inertia, and $ F_{i}, \tau_{i} $ are the applied force and torque, respectively. Define $ \bar{u}_{i}=[F_{i}, \tau_{i}]^{T} $ as the control input.
\begin{figure}[!tb]
	\centering
	\includegraphics[width=0.7\hsize]{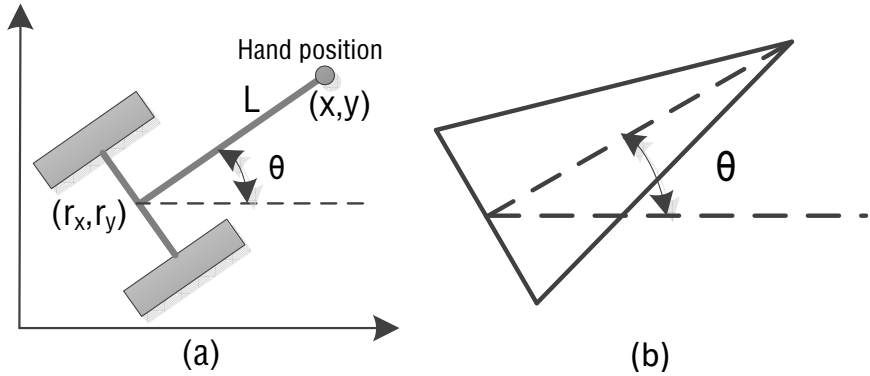}
	\caption{(a) Hand position. (b) The vehicle orientation presentation in Fig.~\ref{snapshots_nonholo}.}
	\label{hand_position}
\end{figure}

\begin{figure*}[!tb]
	\centering
	\includegraphics[width=\hsize]{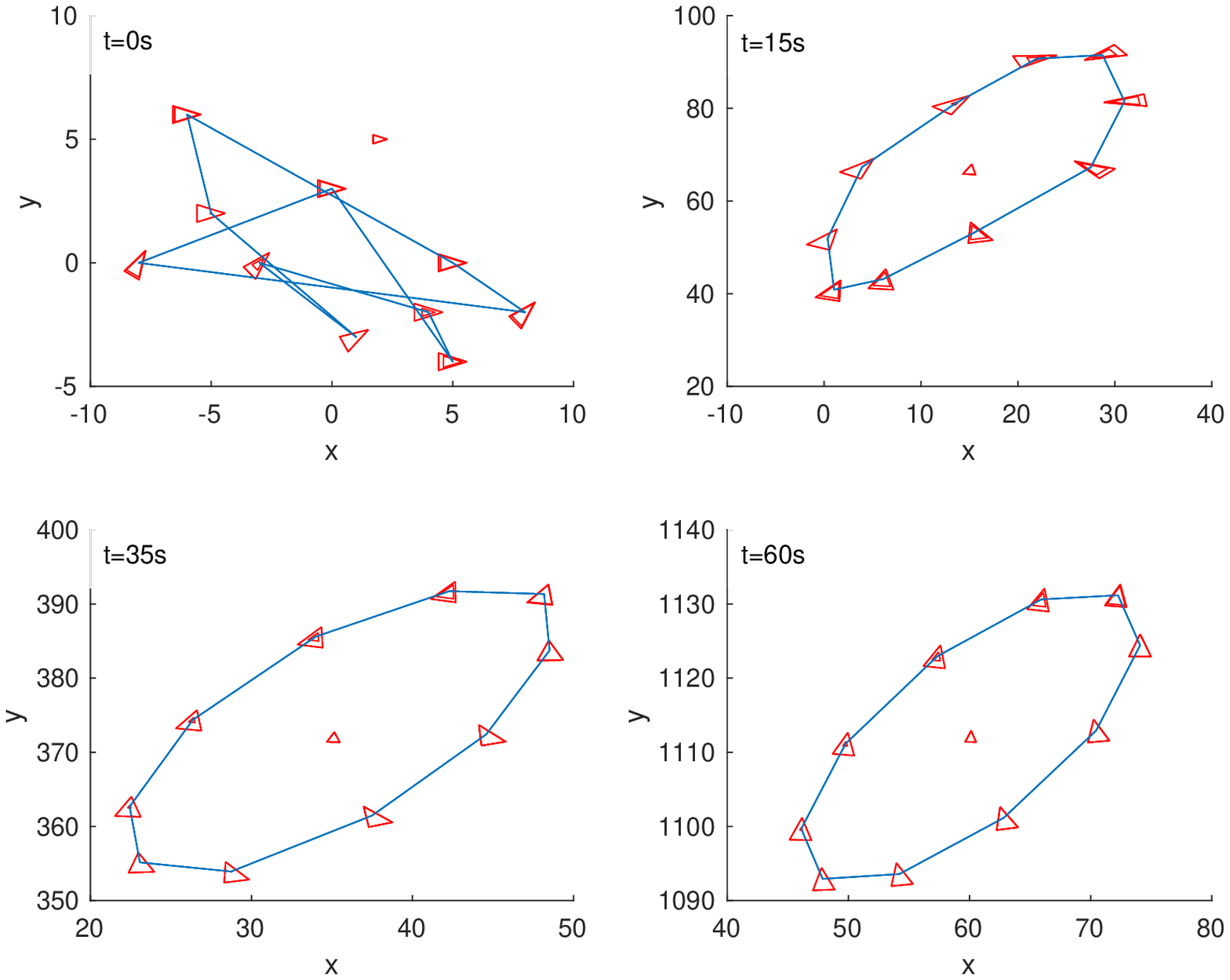}
	\caption{The movement snapshots of nonholonomic vehicles with orientation angles.}
	\label{snapshots_nonholo}
\end{figure*}
As shown in~\cite{lawton_decentralized_2003}, we focus on the TVF control of the vehicle's ``hand" position $ s_{i}=(x_{i}, y_{i}) $ which lies a distance $ L_{i} $ along the line that is normal to the wheel axis and intersects the wheel axis at the center point $ (r_{xi}, r_{yi}) $, as shown in Fig.~\ref{hand_position}a. Therefore, the hand position is defined as
\begin{equation}
\begin{pmatrix}
x_{i}\\y_{i}
\end{pmatrix}=\begin{pmatrix}
r_{xi}\\r_{yi}
\end{pmatrix}+ L_{i}\begin{pmatrix}
cos(\theta_{i})\\sin(\theta_{i})
\end{pmatrix}.
\end{equation}
Using the output feedback linearizing technique~\cite{lawton_decentralized_2003}, we have
\begin{equation}
\begin{aligned}
\bar{u}_{i}=& \begin{pmatrix}
\frac{1}{m_{i}}cos(\theta_{i}) & -\frac{L_{i}}{J_{i}}sin(\theta_{i})\\
\frac{1}{m_{i}}sin(\theta_{i}) & \frac{L_{i}}{J_{i}}cos(\theta_{i})
\end{pmatrix}^{-1}  \left[ u_{i} -\begin{pmatrix}
-\bar{v}_{i}\bar{w}_{i}sin(\theta_{i})-L_{i}\bar{w}_{i}^{2}cos(\theta_{i})\\
\bar{v}_{i}\bar{w}_{i}cos(\theta_{i})-L_{i}\bar{w}_{i}^{2}sin(\theta_{i})
\end{pmatrix} \right]
\end{aligned}
\end{equation}
where $ u_{i} $ is the linearized control input. Then the input output dynamics of each vehicle can be described as a double integrator system
\begin{equation}
\ddot{s}_{i}=u_{i}.
\end{equation}

Define the TVF shape information for followers as
\begin{equation*}
h_{i}(t)=\begin{bmatrix}
-r \, cos(wt+(i-1) \pi /5) + r \, sin(wt+(i-1) \pi /5)\\
2r \, sin(wt+(i-1) \pi /5)\\
w \, r \, cos(wt+(i-1) \pi /5) +w \, r \, sin(wt+(i-1) \pi /5)\\
2w \, r \, cos(wt+(i-1) \pi /5)\\
\end{bmatrix}
\end{equation*}
where $ r=10, w=0.5, i=1,\ldots,10 $. Similar as the Example 1, implement the control input $ u_{i} $ in~\eqref{eq:protocol_tracking_input} to this linearized model with
$$
A =  
\begin{bmatrix}
0 & 0 & 1 & 0\\
0 & 0 & 0 & 1\\
0 & 0 & 0 & 0\\
0 & 0 & 0 & 0
\end{bmatrix}, \quad 
B =\begin{bmatrix}
0 & 0\\
0 & 0\\
1 & 0\\
0 & 1
\end{bmatrix}, \quad 
C =  
\begin{bmatrix}
1 & 0 & 0 & 0\\
0 & 1 & 0 & 0\\
0 & 0 & 1 & 0
\end{bmatrix}.
$$
The leader's bounded input is defined as $ u_{0}(t)=[e^{-t}, \left\vert sin(\frac{t}{2}) \right\vert  ]^{T} $, and $ \beta =4 $. Choose the same parameter $ m_{i}=10.1 \, kg, J_{i}=0.13 \, kg \, m^{2} $ and $ L_{i}=0.12\, m $ as in~\cite{lawton_decentralized_2003}.

The initial positions and the movement snapshots of multi-vehicle system at different time-instants are shown in Fig.~\ref{snapshots_nonholo}. Note here that the orientation angle $ \theta $ of each nonholonomic vehicle is represented as in Fig.~\ref{hand_position}b, and the smaller triangle represents the leader. As we can see, the decagon shape is formed successfully by the ten follower vehicles and rotates around the leader vehicle. The coupling weights $ c_{i} $ associated with follower vehicles are depicted in Fig.~\ref{ci_nonholo}.

\section{conclusion}\label{conclusion}
\indent The unified framework of TVF protocol design based on the distributed observer viewpoint for general linear MAS has demonstrated from undirected to directed topology, from formation stabilization to tracking and from with a leader of no input to with one of bounded input whose information is only available to a small subset of followers.
The proposed observer-type protocols, which don't need any global information, e.g., the smallest positive eigenvalue of Laplacian matrices,  is thus fully distributed.
The absolute and relative output measurements used in this paper are more available in practice compared with the absolute or relative state ones.
Future work will focus on releasing the constraint that agents need to be introspective when dealing with the TVF tracking issue under the directed topology.

\bibliographystyle{model1-num-names}
\bibliography{refrence}

\begin{thebibliography}{43}
\expandafter\ifx\csname natexlab\endcsname\relax\def\natexlab#1{#1}\fi
\providecommand{\bibinfo}[2]{#2}
\ifx\xfnm\relax \def\xfnm[#1]{\unskip,\space#1}\fi
\bibitem[{Olfati-Saber and Murray(2004)}]{olfati-saber_consensus_2004}
\bibinfo{author}{R.~Olfati-Saber}, \bibinfo{author}{R.~M. Murray},
\newblock \bibinfo{title}{Consensus problems in networks of agents with
  switching topology and time-delays},
\newblock \bibinfo{journal}{IEEE Transactions on Automatic Control}
  \bibinfo{volume}{49} (\bibinfo{year}{2004}) \bibinfo{pages}{1520--1533}.
\bibitem[{Ren and Beard(2005)}]{ren_consensus_2005}
\bibinfo{author}{W.~Ren}, \bibinfo{author}{R.~W. Beard},
\newblock \bibinfo{title}{Consensus seeking in multiagent systems under
  dynamically changing interaction topologies},
\newblock \bibinfo{journal}{IEEE Transactions on Automatic Control}
  \bibinfo{volume}{50} (\bibinfo{year}{2005}) \bibinfo{pages}{655--661}.
\bibitem[{Yu et~al.(2017)Yu, Wang, Cheng, Yu, and Wen}]{yu_second-order_2017}
\bibinfo{author}{W.~Yu}, \bibinfo{author}{H.~Wang}, \bibinfo{author}{F.~Cheng},
  \bibinfo{author}{X.~Yu}, \bibinfo{author}{G.~Wen},
\newblock \bibinfo{title}{Second-{Order} {Consensus} in {Multiagent} {Systems}
  via {Distributed} {Sliding} {Mode} {Control}},
\newblock \bibinfo{journal}{IEEE Transactions on Cybernetics}
  \bibinfo{volume}{47} (\bibinfo{year}{2017}) \bibinfo{pages}{1872--1881}.
\bibitem[{Ji et~al.(2008)Ji, Ferrari-Trecate, Egerstedt, and
  Buffa}]{ji_containment_2008}
\bibinfo{author}{M.~Ji}, \bibinfo{author}{G.~Ferrari-Trecate},
  \bibinfo{author}{M.~Egerstedt}, \bibinfo{author}{A.~Buffa},
\newblock \bibinfo{title}{Containment {Control} in {Mobile} {Networks}},
\newblock \bibinfo{journal}{IEEE Transactions on Automatic Control}
  \bibinfo{volume}{53} (\bibinfo{year}{2008}) \bibinfo{pages}{1972--1975}.
\bibitem[{Cheng et~al.(2016)Cheng, Wang, Ren, Hou, and
  Tan}]{cheng_containment_2016}
\bibinfo{author}{L.~Cheng}, \bibinfo{author}{Y.~Wang},
  \bibinfo{author}{W.~Ren}, \bibinfo{author}{Z.~G. Hou},
  \bibinfo{author}{M.~Tan},
\newblock \bibinfo{title}{Containment {Control} of {Multiagent} {Systems}
  {With} {Dynamic} {Leaders} {Based} on a $pi^{\textrm{n}}$ -{Type}
  {Approach}},
\newblock \bibinfo{journal}{IEEE Transactions on Cybernetics}
  \bibinfo{volume}{46} (\bibinfo{year}{2016}) \bibinfo{pages}{3004--3017}.
\bibitem[{Fax and Murray(2004)}]{fax_information_2004}
\bibinfo{author}{J.~A. Fax}, \bibinfo{author}{R.~M. Murray},
\newblock \bibinfo{title}{Information flow and cooperative control of vehicle
  formations},
\newblock \bibinfo{journal}{IEEE Transactions on Automatic Control}
  \bibinfo{volume}{49} (\bibinfo{year}{2004}) \bibinfo{pages}{1465--1476}.
\bibitem[{Tanner et~al.(2004)Tanner, Pappas, and
  Kumar}]{tanner_leader--formation_2004}
\bibinfo{author}{H.~G. Tanner}, \bibinfo{author}{G.~J. Pappas},
  \bibinfo{author}{V.~Kumar},
\newblock \bibinfo{title}{Leader-to-formation stability},
\newblock \bibinfo{journal}{IEEE Transactions on Robotics and Automation}
  \bibinfo{volume}{20} (\bibinfo{year}{2004}) \bibinfo{pages}{443--455}.
\bibitem[{Dong and Hu(2016)}]{dong_time-varying_2016-1}
\bibinfo{author}{X.~Dong}, \bibinfo{author}{G.~Hu},
\newblock \bibinfo{title}{Time-varying formation control for general linear
  multi-agent systems with switching directed topologies},
\newblock \bibinfo{journal}{Automatica} \bibinfo{volume}{73}
  (\bibinfo{year}{2016}) \bibinfo{pages}{47--55}.
\bibitem[{Wang and Xin(2013)}]{wang_integrated_2013}
\bibinfo{author}{J.~Wang}, \bibinfo{author}{M.~Xin},
\newblock \bibinfo{title}{Integrated {Optimal} {Formation} {Control} of
  {Multiple} {Unmanned} {Aerial} {Vehicles}},
\newblock \bibinfo{journal}{IEEE Transactions on Control Systems Technology}
  \bibinfo{volume}{21} (\bibinfo{year}{2013}) \bibinfo{pages}{1731--1744}.
\bibitem[{Ghommam et~al.(2016)Ghommam, Luque-Vega, Castillo-Toledo, and
  Saad}]{ghommam_three-dimensional_2016}
\bibinfo{author}{J.~Ghommam}, \bibinfo{author}{L.~F. Luque-Vega},
  \bibinfo{author}{B.~Castillo-Toledo}, \bibinfo{author}{M.~Saad},
\newblock \bibinfo{title}{Three-dimensional distributed tracking control for
  multiple quadrotor helicopters},
\newblock \bibinfo{journal}{Journal of the Franklin Institute}
  \bibinfo{volume}{353} (\bibinfo{year}{2016}) \bibinfo{pages}{2344--2372}.
\bibitem[{Antonelli et~al.(2014)Antonelli, Arrichiello, Caccavale, and
  Marino}]{antonelli_decentralized_2014}
\bibinfo{author}{G.~Antonelli}, \bibinfo{author}{F.~Arrichiello},
  \bibinfo{author}{F.~Caccavale}, \bibinfo{author}{A.~Marino},
\newblock \bibinfo{title}{Decentralized time-varying formation control for
  multi-robot systems},
\newblock \bibinfo{journal}{The Int'l Journal of Robotics Research}
  \bibinfo{volume}{33} (\bibinfo{year}{2014}) \bibinfo{pages}{1029--1043}.
\bibitem[{Liu and Jiang(2013)}]{liu_distributed_2013}
\bibinfo{author}{T.~Liu}, \bibinfo{author}{Z.-P. Jiang},
\newblock \bibinfo{title}{Distributed formation control of nonholonomic mobile
  robots without global position measurements},
\newblock \bibinfo{journal}{Automatica} \bibinfo{volume}{49}
  (\bibinfo{year}{2013}) \bibinfo{pages}{592--600}.
\bibitem[{Sakurama(2016)}]{sakurama_multi-robot_2016}
\bibinfo{author}{K.~Sakurama},
\newblock \bibinfo{title}{Multi-robot {Formation} {Control} over {Distance}
  {Sensor} {Network}},
\newblock \bibinfo{journal}{IFAC-PapersOnLine} \bibinfo{volume}{49}
  (\bibinfo{year}{2016}) \bibinfo{pages}{198--203}.
\bibitem[{Peng et~al.(2013)Peng, Wen, Rahmani, and
  Yu}]{peng_leaderfollower_2013}
\bibinfo{author}{Z.~Peng}, \bibinfo{author}{G.~Wen},
  \bibinfo{author}{A.~Rahmani}, \bibinfo{author}{Y.~Yu},
\newblock \bibinfo{title}{Leader–follower formation control of nonholonomic
  mobile robots based on a bioinspired neurodynamic based approach},
\newblock \bibinfo{journal}{Robotics and Autonomous Systems}
  \bibinfo{volume}{61} (\bibinfo{year}{2013}) \bibinfo{pages}{988--996}.
\bibitem[{Zhang and Liu(2016)}]{zhang_cooperative_2016}
\bibinfo{author}{M.~Zhang}, \bibinfo{author}{H.~H.~T. Liu},
\newblock \bibinfo{title}{Cooperative {Tracking} a {Moving} {Target} {Using}
  {Multiple} {Fixed}-wing {UAVs}},
\newblock \bibinfo{journal}{J Intell Robot Syst} \bibinfo{volume}{81}
  (\bibinfo{year}{2016}) \bibinfo{pages}{505--529}.
\bibitem[{Nigam et~al.(2012)Nigam, Bieniawski, Kroo, and
  Vian}]{nigam_control_2012}
\bibinfo{author}{N.~Nigam}, \bibinfo{author}{S.~Bieniawski},
  \bibinfo{author}{I.~Kroo}, \bibinfo{author}{J.~Vian},
\newblock \bibinfo{title}{Control of {Multiple} {UAVs} for {Persistent}
  {Surveillance}: {Algorithm} and {Flight} {Test} {Results}},
\newblock \bibinfo{journal}{IEEE Transactions on Control Systems Technology}
  \bibinfo{volume}{20} (\bibinfo{year}{2012}) \bibinfo{pages}{1236--1251}.
\bibitem[{Lewis and Tan(1997)}]{lewis_high_1997}
\bibinfo{author}{M.~A. Lewis}, \bibinfo{author}{K.-H. Tan},
\newblock \bibinfo{title}{High {Precision} {Formation} {Control} of {Mobile}
  {Robots} {Using} {Virtual} {Structures}},
\newblock \bibinfo{journal}{Autonomous Robots} \bibinfo{volume}{4}
  (\bibinfo{year}{1997}) \bibinfo{pages}{387--403}.
\bibitem[{Das et~al.(2002)Das, Fierro, Kumar, Ostrowski, Spletzer, and
  Taylor}]{das_vision-based_2002}
\bibinfo{author}{A.~K. Das}, \bibinfo{author}{R.~Fierro},
  \bibinfo{author}{V.~Kumar}, \bibinfo{author}{J.~P. Ostrowski},
  \bibinfo{author}{J.~Spletzer}, \bibinfo{author}{C.~J. Taylor},
\newblock \bibinfo{title}{A vision-based formation control framework},
\newblock \bibinfo{journal}{IEEE Transactions on Robotics and Automation}
  \bibinfo{volume}{18} (\bibinfo{year}{2002}) \bibinfo{pages}{813--825}.
\bibitem[{Balch and Arkin(1998)}]{balch_behavior-based_1998}
\bibinfo{author}{T.~Balch}, \bibinfo{author}{R.~C. Arkin},
\newblock \bibinfo{title}{Behavior-based formation control for multirobot
  teams},
\newblock \bibinfo{journal}{IEEE Transactions on Robotics and Automation}
  \bibinfo{volume}{14} (\bibinfo{year}{1998}) \bibinfo{pages}{926--939}.
\bibitem[{Tanner(2004)}]{tanner_controllability_2004}
\bibinfo{author}{H.~G. Tanner},
\newblock \bibinfo{title}{On the controllability of nearest neighbor
  interconnections},
\newblock in: \bibinfo{booktitle}{Decision and {Control}, 2004. {CDC}. 43rd
  {IEEE} {Conference} on}, volume~\bibinfo{volume}{3},
  \bibinfo{publisher}{IEEE}, \bibinfo{year}{2004}, pp.
  \bibinfo{pages}{2467--2472}.
\bibitem[{Peymani et~al.(2014)Peymani, Grip, Saberi, Wang, and
  Fossen}]{peymani_almost_2014}
\bibinfo{author}{E.~Peymani}, \bibinfo{author}{H.~F. Grip},
  \bibinfo{author}{A.~Saberi}, \bibinfo{author}{X.~Wang},
  \bibinfo{author}{T.~I. Fossen},
\newblock \bibinfo{title}{$ \mathcal{H}_{\infty} $ almost output
  synchronization for heterogeneous networks of introspective agents under
  external disturbances},
\newblock \bibinfo{journal}{Automatica} \bibinfo{volume}{50}
  (\bibinfo{year}{2014}) \bibinfo{pages}{1026--1036}.
\bibitem[{Oh et~al.(2015)Oh, Park, and Ahn}]{oh_survey_2015}
\bibinfo{author}{K.-K. Oh}, \bibinfo{author}{M.-C. Park},
  \bibinfo{author}{H.-S. Ahn},
\newblock \bibinfo{title}{A survey of multi-agent formation control},
\newblock \bibinfo{journal}{Automatica} \bibinfo{volume}{53}
  (\bibinfo{year}{2015}) \bibinfo{pages}{424--440}.
\bibitem[{Dong et~al.(2016)Dong, Li, Ren, and Zhong}]{dong_output_2016}
\bibinfo{author}{X.~Dong}, \bibinfo{author}{Q.~Li}, \bibinfo{author}{Z.~Ren},
  \bibinfo{author}{Y.~Zhong},
\newblock \bibinfo{title}{Output formation-containment analysis and design for
  general linear time-invariant multi-agent systems},
\newblock \bibinfo{journal}{Journal of the Franklin Institute}
  \bibinfo{volume}{353} (\bibinfo{year}{2016}) \bibinfo{pages}{322--344}.
\bibitem[{Rahimi et~al.(2014)Rahimi, Abdollahi, and
  Naqshi}]{rahimi_time-varying_2014}
\bibinfo{author}{R.~Rahimi}, \bibinfo{author}{F.~Abdollahi},
  \bibinfo{author}{K.~Naqshi},
\newblock \bibinfo{title}{Time-varying formation control of a collaborative
  heterogeneous multi agent system},
\newblock \bibinfo{journal}{Robotics and Autonomous Systems}
  \bibinfo{volume}{62} (\bibinfo{year}{2014}) \bibinfo{pages}{1799--1805}.
\bibitem[{Dong et~al.(2016)Dong, Xiang, Han, Li, and
  Ren}]{dong_distributed_2016}
\bibinfo{author}{X.~Dong}, \bibinfo{author}{J.~Xiang},
  \bibinfo{author}{L.~Han}, \bibinfo{author}{Q.~Li}, \bibinfo{author}{Z.~Ren},
\newblock \bibinfo{title}{Distributed {Time}-{Varying} {Formation} {Tracking}
  {Analysis} and {Design} for {Second}-{Order} {Multi}-{Agent} {Systems}},
\newblock \bibinfo{journal}{J Intell Robot Syst}  (\bibinfo{year}{2016})
  \bibinfo{pages}{1--13}.
\bibitem[{Wang et~al.(2017)Wang, Dong, Li, and Ren}]{wang2017distributed}
\bibinfo{author}{R.~Wang}, \bibinfo{author}{X.~Dong}, \bibinfo{author}{Q.~Li},
  \bibinfo{author}{Z.~Ren},
\newblock \bibinfo{title}{Distributed time-varying formation control for linear
  swarm systems with switching topologies using an adaptive output-feedback
  approach},
\newblock \bibinfo{journal}{IEEE Transactions on Systems, Man, and Cybernetics:
  Systems}  (\bibinfo{year}{2017}).
\bibitem[{Jadbabaie et~al.(2003)Jadbabaie, Lin, and
  Morse}]{jadbabaie_coordination_2003}
\bibinfo{author}{A.~Jadbabaie}, \bibinfo{author}{J.~Lin},
  \bibinfo{author}{A.~S. Morse},
\newblock \bibinfo{title}{Coordination of groups of mobile autonomous agents
  using nearest neighbor rules},
\newblock \bibinfo{journal}{IEEE Transactions on Automatic Control}
  \bibinfo{volume}{48} (\bibinfo{year}{2003}) \bibinfo{pages}{988--1001}.
\bibitem[{Ren(2007)}]{ren_consensus_2007}
\bibinfo{author}{W.~Ren},
\newblock \bibinfo{title}{Consensus strategies for cooperative control of
  vehicle formations},
\newblock \bibinfo{journal}{IET Control Theory \& Applications}
  \bibinfo{volume}{1} (\bibinfo{year}{2007}) \bibinfo{pages}{505--512}.
\bibitem[{Li et~al.(2011)Li, Liu, Lin, and Ren}]{li_consensus_2011}
\bibinfo{author}{Z.~Li}, \bibinfo{author}{X.~Liu}, \bibinfo{author}{P.~Lin},
  \bibinfo{author}{W.~Ren},
\newblock \bibinfo{title}{Consensus of linear multi-agent systems with
  reduced-order observer-based protocols},
\newblock \bibinfo{journal}{Systems \& Control Letters} \bibinfo{volume}{60}
  (\bibinfo{year}{2011}) \bibinfo{pages}{510--516}.
\bibitem[{Cao et~al.(2013)Cao, Yu, Ren, and Chen}]{cao_overview_2013}
\bibinfo{author}{Y.~Cao}, \bibinfo{author}{W.~Yu}, \bibinfo{author}{W.~Ren},
  \bibinfo{author}{G.~Chen},
\newblock \bibinfo{title}{An {Overview} of {Recent} {Progress} in the {Study}
  of {Distributed} {Multi}-{Agent} {Coordination}},
\newblock \bibinfo{journal}{IEEE Transactions on Industrial Informatics}
  \bibinfo{volume}{9} (\bibinfo{year}{2013}) \bibinfo{pages}{427--438}.
\bibitem[{Jiang et~al.(2017)Jiang, Wen, Meng, and
  Rahmani}]{jiang2017distributed}
\bibinfo{author}{W.~Jiang}, \bibinfo{author}{G.~Wen},
  \bibinfo{author}{Y.~Meng}, \bibinfo{author}{A.~Rahmani},
\newblock \bibinfo{title}{Distributed adaptive time-varying formation tracking
  for linear multi-agent systems: A dynamic output approach},
\newblock in: \bibinfo{booktitle}{Control Conference (CCC), 2017 36th Chinese},
  \bibinfo{organization}{IEEE}, pp. \bibinfo{pages}{8571--8576}.
\bibitem[{Godsil and Royle(2001)}]{godsil_algebraic_2001}
\bibinfo{author}{C.~Godsil}, \bibinfo{author}{G.~F. Royle},
  \bibinfo{title}{Algebraic {Graph} {Theory}}, \bibinfo{publisher}{Springer
  Science \& Business Media}, \bibinfo{address}{New York, NY, USA},
  \bibinfo{year}{2001}.
\bibitem[{Qu(2009)}]{qu_cooperative_2009}
\bibinfo{author}{Z.~Qu}, \bibinfo{title}{Cooperative {Control} of {Dynamical}
  {Systems}: {Applications} to {Autonomous} {Vehicles}},
  \bibinfo{publisher}{Springer Science \& Business Media},
  \bibinfo{year}{2009}.
\bibitem[{Mei et~al.(2014)Mei, Ren, and Chen}]{mei_consensus_2014}
\bibinfo{author}{J.~Mei}, \bibinfo{author}{W.~Ren}, \bibinfo{author}{J.~Chen},
\newblock \bibinfo{title}{Consensus of second-order heterogeneous multi-agent
  systems under a directed graph},
\newblock in: \bibinfo{booktitle}{2014 {American} {Control} {Conference}}, pp.
  \bibinfo{pages}{802--807}.
\bibitem[{Bernstein(2009)}]{bernstein_matrix_2009}
\bibinfo{author}{D.~S. Bernstein}, \bibinfo{title}{Matrix {Mathematics}:
  {Theory}, {Facts}, and {Formulas}.}, \bibinfo{publisher}{Princeton University
  Press}, \bibinfo{year}{2009}.
\bibitem[{Wen et~al.(2016)Wen, Zhao, Duan, Yu, and Chen}]{wen_containment_2016}
\bibinfo{author}{G.~Wen}, \bibinfo{author}{Y.~Zhao}, \bibinfo{author}{Z.~Duan},
  \bibinfo{author}{W.~Yu}, \bibinfo{author}{G.~Chen},
\newblock \bibinfo{title}{Containment of {Higher}-{Order} {Multi}-{Leader}
  {Multi}-{Agent} {Systems}: {A} {Dynamic} {Output} {Approach}},
\newblock \bibinfo{journal}{IEEE Transactions on Automatic Control}
  \bibinfo{volume}{61} (\bibinfo{year}{2016}) \bibinfo{pages}{1135--1140}.
\bibitem[{Li et~al.(2013)Li, Ren, Liu, and Fu}]{li_consensus_2013}
\bibinfo{author}{Z.~Li}, \bibinfo{author}{W.~Ren}, \bibinfo{author}{X.~Liu},
  \bibinfo{author}{M.~Fu},
\newblock \bibinfo{title}{Consensus of {Multi}-{Agent} {Systems} {With}
  {General} {Linear} and {Lipschitz} {Nonlinear} {Dynamics} {Using}
  {Distributed} {Adaptive} {Protocols}},
\newblock \bibinfo{journal}{IEEE Transactions on Automatic Control}
  \bibinfo{volume}{58} (\bibinfo{year}{2013}) \bibinfo{pages}{1786--1791}.
\bibitem[{Lv et~al.(2016)Lv, Li, Duan, and Feng}]{lv_novel_2016}
\bibinfo{author}{Y.~Lv}, \bibinfo{author}{Z.~Li}, \bibinfo{author}{Z.~Duan},
  \bibinfo{author}{G.~Feng},
\newblock \bibinfo{title}{Novel distributed robust adaptive consensus protocols
  for linear multi-agent systems with directed graphs and external
  disturbances},
\newblock \bibinfo{journal}{International Journal of Control}
  \bibinfo{volume}{0} (\bibinfo{year}{2016}) \bibinfo{pages}{1--11}.
\bibitem[{Hong et~al.(2006)Hong, Hu, and Gao}]{hong_tracking_2006}
\bibinfo{author}{Y.~Hong}, \bibinfo{author}{J.~Hu}, \bibinfo{author}{L.~Gao},
\newblock \bibinfo{title}{Tracking control for multi-agent consensus with an
  active leader and variable topology},
\newblock \bibinfo{journal}{Automatica} \bibinfo{volume}{42}
  (\bibinfo{year}{2006}) \bibinfo{pages}{1177--1182}.
\bibitem[{Krstic et~al.(1995)Krstic, Kokotovic, and
  Kanellakopoulos}]{krstic_nonlinear_1995}
\bibinfo{author}{M.~Krstic}, \bibinfo{author}{P.~V. Kokotovic},
  \bibinfo{author}{I.~Kanellakopoulos}, \bibinfo{title}{Nonlinear and
  {Adaptive} {Control} {Design}}, \bibinfo{publisher}{John Wiley \& Sons,
  Inc.}, \bibinfo{address}{New York, NY, USA}, \bibinfo{edition}{1st} edition,
  \bibinfo{year}{1995}.
\bibitem[{Khalil(1996)}]{khalil1996noninear}
\bibinfo{author}{H.~K. Khalil},
\newblock \bibinfo{title}{Noninear systems},
\newblock \bibinfo{journal}{Prentice-Hall, New Jersey} \bibinfo{volume}{2}
  (\bibinfo{year}{1996}) \bibinfo{pages}{5--1}.
\bibitem[{Lv et~al.(2015)Lv, Li, Duan, and Chen}]{lv_fully_2015}
\bibinfo{author}{Y.~Lv}, \bibinfo{author}{Z.~Li}, \bibinfo{author}{Z.~Duan},
  \bibinfo{author}{J.~Chen},
\newblock \bibinfo{title}{Fully {Distributed} {Adaptive} {Output} {Feedback}
  {Protocols} for {Linear} {Multi}-{Agent} {Systems} with {Directed} {Graphs}:
  {A} {Sequential} {Observer} {Design} {Approach}},
\newblock \bibinfo{journal}{arXiv:1511.01297 [cs]}  (\bibinfo{year}{2015}).
\bibitem[{Lawton et~al.(2003)Lawton, Beard, and
  Young}]{lawton_decentralized_2003}
\bibinfo{author}{J.~R.~T. Lawton}, \bibinfo{author}{R.~W. Beard},
  \bibinfo{author}{B.~J. Young},
\newblock \bibinfo{title}{A decentralized approach to formation maneuvers},
\newblock \bibinfo{journal}{IEEE Transactions on Robotics and Automation}
  \bibinfo{volume}{19} (\bibinfo{year}{2003}) \bibinfo{pages}{933--941}.

\end{thebibliography}







\end{document}